\documentclass[a4paper,USenglish,cleveref, autoref]{lipics-v2019}

\usepackage{csquotes}

\usepackage{tabularx}
\usepackage{booktabs}

\usepackage{amsmath, amssymb, amsthm, mathtools, thmtools}
\usepackage{nicefrac}

\usepackage{xcolor}

\definecolor{mred}{RGB}{215,25,28}
\definecolor{mdarkblue}{RGB}{44,123,182}
\definecolor{mlightblue}{RGB}{171,217,233}
\definecolor{morange}{RGB}{253,174,97}
\definecolor{color1}{RGB}{215,25,28}
\definecolor{color2}{RGB}{44,123,182}
\definecolor{color3}{RGB}{253,174,97}

\usepackage{tikz}
\usetikzlibrary{decorations.pathreplacing}
\usetikzlibrary{math}
\usetikzlibrary{shapes,snakes}
\usetikzlibrary{patterns}

\usepackage{xspace}

\newcommand{\problem}[1]{\textnormal{#1}\xspace}

\newcommand{\rar}[1]{\textnormal{RAR}$(#1)$}
\newcommand{\lrs}[1]{\textnormal{LRS}$(#1)$}

\newcommand{\rai}{\problem{RAI}}
\newcommand{\threesat}{\problem{$3$-SAT}}
\newcommand{\tailoredsat}{\problem{$3$-SAT$^*$}}
\newcommand{\oneinthreesat}{\problem{$1$-in-$3$-SAT}}
\newcommand{\threedm}{\problem{$3$-DM}}
\newcommand{\tailoredmatching}{\problem{$3$-DM$^*$}}

\newcommand{\jobs}{\mathcal{J}}
\newcommand{\machs}{\mathcal{M}}
\newcommand{\emachs}{\machs}
\newcommand{\ress}{\mathcal{R}}

\newcommand{\makespan}{C_{\max}}

\newcommand{\true}{\top}
\newcommand{\false}{\bot}

\newcommand{\clausemach}{\mathtt{CMach}}
\newcommand{\truthmach}{\mathtt{TMach}}

\newcommand{\clausejob}{\mathtt{CJob}}
\newcommand{\truthjob}{\mathtt{TJob}}
\newcommand{\variablejob}{\mathtt{VJob}}

\newcommand{\gatemach}{\mathtt{GMach}}
\newcommand{\bridgein}{\mathtt{BMachIn}}
\newcommand{\bridgeout}{\mathtt{BMachOut}}
\newcommand{\bridgejob}{\mathtt{BJob}}
\newcommand{\highwayjob}{\mathtt{HJob}}

\newcommand{\tblock}{\mathcal{T}}
\newcommand{\pblock}{\mathcal{P}}
\newcommand{\sblock}{\mathcal{S}}
\newcommand{\cblock}{\mathcal{C}}

\newcommand{\eps}{\varepsilon}

\newcommand{\NN}{\mathbb{N}}

\DeclarePairedDelimiter\ceil{\lceil}{\rceil}

\DeclarePairedDelimiter\set{\lbrace}{\rbrace}
\DeclarePairedDelimiterX\sett[2]{\lbrace}{\rbrace}{ #1 \,\delimsize| \,\mathopen{} #2 }

\title{Inapproximability Results for Scheduling with Interval and Resource Restrictions} 

\author{Marten Maack}{Department of Computer Science, Kiel University,  Kiel, Germany}{mmaa@informatik.uni-kiel.de}{}{}

\author{Klaus Jansen}{Department of Computer Science, Kiel University,  Kiel, Germany}{kj@informatik.uni-kiel.de}{}{}

\authorrunning{Marten Maack and Klaus Jansen}



\category{}

\relatedversion{}

\supplement{}

\funding{This work was partially supported by the German Academic Exchange Service (DAAD) and by the German Research Foundation (DFG) project JA 612/15-2.}

\acknowledgements{We thank Malin Rau and Lars Rohwedder for helpfull discussions on the problem.}

\nolinenumbers 

\hideLIPIcs  

\EventEditors{John Q. Open and Joan R. Access}
\EventNoEds{2}
\EventLongTitle{42nd Conference on Very Important Topics (CVIT 2016)}
\EventShortTitle{CVIT 2016}
\EventAcronym{CVIT}
\EventYear{2016}
\EventDate{December 24--27, 2016}
\EventLocation{Little Whinging, United Kingdom}
\EventLogo{}
\SeriesVolume{42}
\ArticleNo{23}

\begin{document}

\maketitle

\begin{abstract}
In the restricted assignment problem, the input consists of a set of machines and a set of jobs each with a processing time and a subset of eligible machines.
The goal is to find an assignment of the jobs to the machines minimizing the makespan, that is, the maximum summed up processing time any machine receives.
Herein, jobs should only be assigned to those machines on which they are eligible.
It is well-known that there is no polynomial time approximation algorithm with an approximation guarantee of less than 1.5 for the restricted assignment problem unless P=NP.
In this work, we show hardness results for variants of the restricted assignment problem with particular types of restrictions.

For the case of interval restrictions---where the machines can be totally ordered such that jobs are eligible on consecutive machines---we show that there is no polynomial time approximation scheme (PTAS) unless P=NP.
The question of whether a PTAS for this variant exists was stated as an open problem before, and PTAS results for special cases of this variant are known.

Furthermore, we consider a variant with resource restriction where the sets of eligible machines are of the following form:
There is a fixed number of (renewable) resources, each machine has a capacity, and each job a demand for each resource. 
A job is eligible on a machine if its demand is at most as big as the capacity of the machine for each resource.
For one resource, this problem has been intensively studied under several different names and is known to admit a PTAS, and for two resources the variant with interval restrictions is contained as a special case. 
Moreover, the version with multiple resources is closely related to makespan minimization on parallel machines with a low rank processing time matrix.
We show that there is no polynomial time approximation algorithm with a rate smaller than $48/47 \approx 1.02$ or $1.5$ for scheduling with resource restrictions with $2$ or $4$ resources, respectively, unless P$=$NP.
All our results can be extended to the so called Santa Claus variants of the problems where the goal is to maximize the minimal processing time any machine receives.
\end{abstract}

\section{Introduction}
\label{sec:introduction}

Consider the restricted assignment problem:
Given a set of machines $\machs$ and a set of jobs $\jobs$ each with a processing time or size $p_{j}$ and a subset of eligible machines $\emachs(j)\subseteq\machs$, the goal is to find a schedule $\sigma:\jobs\rightarrow\machs$ with $\sigma(j)\in\emachs(j)$ for each job $j$ and minimizing the makespan $\makespan(\sigma) = \max_{i\in\machs}\sum_{j\in\sigma^{-1}(i)}p_{j}$.

In a seminal work, Lenstra, Shmoys and Tardos \cite{LenstraST90} presented a $2$-approximation for restricted assignment and also showed that there is no polynomial time approximation algorithm with rate smaller than $1.5$ for the problem, unless P=NP.
Closing this gap is a prominent open problem in approximation and scheduling theory \cite{SchuurmanW99,WS11book}.
If there are no restrictions, i.e., $\emachs(j)=\machs$ for each job $j$, we have the classical problem of makespan minimization on identical parallel machines (machine scheduling) which is already strongly NP-hard.
On the other hand, machine scheduling is well-known to admit a polynomial time approximation scheme (PTAS) due to a classical result by Hochbaum and Shmoys~\cite{HochbaumS87identische}.
In recent years, the approximability of special cases of restricted assignment has been intensively studied (see, e.g., \cite{ChakrabartyS16,EbenlendrKS14,HuangO16,JansenR18GrapbBalancing}) with one line of research focusing on the existence of approximation schemes (see, e.g., \cite{EpsteinL11,JansenMS17,MuratoreSW10,OLL08}).
The present work seeks to contribute in this research direction.

\subparagraph{Interval Restrictions.}

Arguably one of the most natural variants of the restricted assignment problem is the case of scheduling with interval restrictions (\rai).
In this variant, the machines are totally ordered and each job is eligible on consecutive machines.
More precisely, we have $\machs = \set{M_1,\dots,M_m}$, and for each job $j$ we have $\emachs(j) = \set{M_\ell,\dots,M_r}$ for some indices $\ell,r\in[m]$.
Several special cases of \rai are known to admit a PTAS: the hierarchical case~\cite{OLL08}, where for each job the interval of eligible machines starts with the first machine; the nested case \cite{MuratoreSW10,EpsteinL11}, where $\emachs(j)\subseteq\emachs(j')$, $\emachs(j')\subseteq\emachs(j)$ or $\emachs(j)\cap\emachs(j')=\emptyset$ for each pair of jobs $(j,j')$; and the inclusion-free case \cite{Schwarz10Thesis,KhodamoradiKRS16}, where $\emachs(j)\subseteq\emachs(j')$ implies that $j$ and $j'$ share either their first or last eligible machine.
Furthermore, for general \rai, a $2-2/(\max_{j\in\jobs}p_j)$-approximation due to Schwarz \cite{Schwarz10Thesis} is known (assuming integral processing times); and the special case with two distinct processing times is even polynomial time solvable \cite{WangS16}.
Note that the problem has also been studied in the context of online algorithms (see \cite{LeeLP13survey,LeungL16update}).

The question of whether there is a PTAS for \rai has been posed by several authors \cite{Khodamoradi16Thesis,Schwarz10Thesis,WangS16}.
As the main result of the present work, we resolve this question in the negative:
\begin{theorem}\label{thm:result_rai}
There is no PTAS for scheduling with interval restrictions unless P$=$NP.
\footnote{There is a paper \cite{KhodamoradiKRS13} claiming to have found a PTAS for \rai. 
However, according to \cite{PTASflawed}, the result is not correct and the authors published a revised version of the paper \cite{KhodamoradiKRS16} claiming a less general result, namely, a PTAS for the inclusion-free case.
}
\end{theorem}

\subparagraph{Resource Restrictions.}

The second variant considered in this work, is the problem of scheduling with resource restrictions with $R$ resources (\rar{R}).
Herein, a set $\ress$ or $R$ (renewable) resources is given, each machine $i$ has a resource capacity $c_{r}(i)$ and each job $j$ has a resource demand $d_{r}(j)$ for each $r\in\ress$.
Job $j$ is eligible on machine $i$ if $d_{r}(j)\leq c_{r}(i)$ for each resource~$r$.
For $R=1$, the problem is equivalent to the mentioned hierarchical case and has been studied intensively \cite{LeungL08survey,LeungL16update}.
Furthermore, it is not hard to see that \rai is properly placed between \rar{1} and \rar{2} (see \cref{sec:resource}) and hence there is a close relationship between the two problems.
For arbitrary $R$, the problem was mentioned in a work by Bhaskara et el.~\cite{BhaskaraKTW13} under the name of geometrically restricted scheduling\footnote{The demands $d(j)$ and capacities $c(i)$ may be interpreted as points in $R$-dimensional space.} but to the best of our knowledge it has not been further studied up to now.
There is, however, a close relationship to the low rank version of makespan minimization on unrelated parallel machines (unrelated scheduling) introduced in \cite{BhaskaraKTW13}.
In the problem of unrelated scheduling, the processing time of each job is dependent on the machine it is scheduled on, that is, a processing time matrix $(p_{ij})_{j\in\jobs,i\in\machs}$ is given in the input.
Unrelated scheduling is a classical problem, and the 2-approximation by Lensta et al. \cite{LenstraST90} was actually formulated for this problem.
Restricted assignment can be seen as a special case of unrelated scheduling by setting $p_{ij} = p_j$ for $i\in\emachs(j)$ and $p_{ij} = \infty$ otherwise.
In the rank $D$ version of unrelated scheduling (\lrs{D}), the processing time matrix has a rank of at most $D$, or, equivalently \cite{CMYZ17}, we may assume that there are $D$-dimensional size vectors $s(j)$ for each job $j$ and speed vectors $v(i)$ for each machine $i$ such that $p_{ij} = \sum_{k=1}^{D}s_k(j)\cdot v_k(i)$.
Considering the latter definition, scheduling with resource restrictions may intuitively be seen as the restricted assignment equivalent of low rank unrelated scheduling.
It is not hard to see that formally already for \rar{1} instances the processing time matrix can have rank $|\machs|$.
However, \lrs{D} includes approximations of any \rar{D-1} instance with arbitrary precision (see \cref{sec:resource} for details).
The case with $D=1$ of \lrs{D} is equivalent with the classical problem of makespan minimization on uniformly related machines and well known to admit a PTAS \cite{HochbaumS88uniforme}.
Bhaskara et el.~\cite{BhaskaraKTW13} gave a quasi-polynomial time approximation scheme (QPTAS) for $D=2$, and showed that there is no PTAS or approximation algorithm with rate smaller than $1.5$ for $D\geq 4$ or $D\geq 7$ respectively.
The latter two results have been improved from $D=4$ to $D=3$ by Chen et al. \cite{CMYZ17} and from $D=7$ to $D=4$ by Chen, Ye and Zhang \cite{ChenYZ14}.  
We present similar inapproximability results for scheduling with resource restrictions:
\begin{theorem}\label{thm:result_rar}
There is no approximation algorithm with rate less than $48/47 \approx 1.02$ or $1.5$ for scheduling with resource restrictions with $2$ or $4$ resources, respectively, unless P$=$NP.
\end{theorem}

\subparagraph*{Santa Claus.}

The problems of restricted assignment and unrelated scheduling are also studied with the reverse objective of maximizing the minimal machine load $\min_{i\in\machs}\sum_{j\in\sigma^{-1}(i)}p_{ij}$.
Usually these variants are described in a more game theoretical context with players instead of machines, goods instead of jobs, and values instead of processing times, and sometimes unrelated scheduling with the reverse objective is called the Santa Claus problem.
In this paper, we will mostly stick to the scheduling notation but denote the variants of the considered problems with reverse objective as the Santa Claus version of the respective problem.

For the Santa Claus version of the restricted assignment problem a 13-approximation due to Annamalai, Kalaitzis and Svensson \cite{AnnamalaiKS17} is known, which has been improved to a rate of $6+\eps$ by both Cheng and Mao \cite{ChengM18} and Davies, Rothvoss and Zhang \cite{DRZ18}.
PTAS results are known for the case without restrictions \cite{Woeginger97} and the inclusion-free interval case \cite{KhodamoradiKRS16}.

Our results can be directly transferred to the Santa Claus versions of the respective problems:
\begin{theorem}\label{thm:results_santa}
Unless P$=$NP, there is no PTAS for the Santa Claus version of scheduling with interval restrictions and no approximation algorithm with rate less than $47/46$ or $2$ for the Santa Claus version of scheduling with resource restrictions with $2$ or $4$ resources, respectively.
\end{theorem}

\subparagraph{Paper structure.}

In the remainder of this section, we discuss some further related literature and present preliminary considerations needed throughout the paper. 
In \cref{sec:interval}, we present our results for \rai; in \cref{sec:resource}, we discuss the problem of \rar{R}; and lastly, in \cref{sec:conclusion}, we present some open problems and possible future research directions.

\subparagraph*{Further Related Work.}

First note that if the number of machines is constant, there is a fully polynomial time approximation scheme (FPTAS) already for unrelated scheduling~\cite{HorowitzS76}.
Furthermore, for some broad overview concerning parallel machine scheduling with different kinds of restrictions in the context of online and approximation algorithms, we refer to the surveys by Lee et al. \cite{LeeLP13survey} and Leung and Li \cite{LeungL08survey,LeungL16update}.

We already discussed many variants of restricted assignment that admit a PTAS.
In particular, Ou, Leung and Li \cite{OLL08} presented a PTAS for the hierarchical case; Epstein and Levin \cite{EpsteinL11} and Muratore, Schwarz and Woeginger \cite{MuratoreSW10} for the nested case; and Schwarz \cite{Schwarz10Thesis} and Khodamoradi et al. \cite{KhodamoradiKRS16} for the inclusion-free case.
Another case that has been studied in the literature is the tree-hierarchical case, where the machines can be arranged in a rooted tree such that for each job the set of eligible machines corresponds to a path starting at the root. 
It was shown to admit a PTAS by Epstein and Levin \cite{EpsteinL11} and Schwarz \cite{Schwarz10}.
It is not hard to see that all of the above cases contain the hierarchical case as a subcase, and that the tree-hierarchical, nested and inclusion-free case are distinct.
There is, however, a variant admitting a PTAS that covers both the nested and the tree-hierarchical case:
For each instance of the restricted assignment problem the corresponding incidence graph is a bipartite graph whose nodes are given by the jobs and machines and a job $j$ is adjacent to a machine $i$ if $j$ is eligible on $i$. 
Jansen, Maack and Solis-Oba \cite{JansenMS17} showed that there is PTAS for restricted assignment for the case that the clique- or rank-width of the incidence graph is constant.
Furthermore, if the incidence graph is a bi-cograph the clique-width is well-known to be small and this case covers the nested and tree-hierarchical case.
The inclusion-free case, on the other hand, is equivalent to the case that the incidence graph is a bipartite permutation graph \cite{KhodamoradiKRS16} which does not have a bounded clique-width \cite{BrandstadtL03}.
Note that \rar{1} or \rai are equivalent to the cases that the incidence graph is a chain \cite{HeggernesK07} or convex graph \cite{KhodamoradiKRS13}, respectively.
For an overview of the discussed cases, we refer to \cref{fig:problems}.
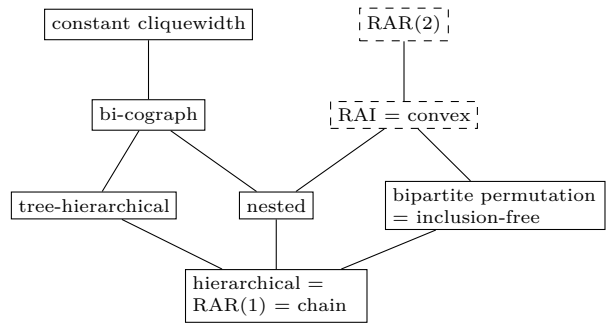
\begin{figure}
\centering
\scriptsize
\begin{tikzpicture}[scale = 0.8]
\pgfmathsetmacro{\h}{1.5}
\pgfmathsetmacro{\w}{3}

\node[draw,rectangle,text width=2.2cm] at (0, 0) (h) {hierarchical = \rar{1} = chain};
\node[draw,rectangle] at (-\w, \h) (th) {tree-hierarchical};
\node[draw,rectangle] at (0, \h) (n) {nested};
\node[draw,rectangle,text width=2.7cm] at (1.2*\w, \h) (if) {bipartite permutation = inclusion-free};
\node[draw,rectangle] at (-0.7*\w, 2*\h) (bc) {bi-cograph};
\node[draw,dashed,rectangle] at (0.7*\w, 2*\h) (rai) {\rai = convex};
\node[draw,rectangle] at (-0.7*\w, 3*\h) (cc) {constant cliquewidth};
\node[draw, dashed,rectangle] at (0.7*\w, 3*\h) (rar) {\rar{2}};

\draw (h) -- (th);
\draw (h) -- (n);
\draw (h) -- (if);
\draw (th) -- (bc);
\draw (n) -- (bc);
\draw (n) -- (rai);
\draw (if) -- (rai);
\draw (rai) -- (rar);
\draw (bc) -- (cc);

\end{tikzpicture}
\caption{An overview of the inclusion structure of several of the discussed variants of restricted assignment. If two problems are connected, the upper includes the lower one. The dashed problems do not admit a PTAS, unless P=NP (see \cref{thm:result_rai}), the remaining ones do.}
\label{fig:problems}
\end{figure}

Lastly, there has been a series of promising results in recent years concerning restricted assignment and variants thereof, and we highlight a few of them.
In a breakthrough result, Svensson \cite{Svensson12} showed that a certain integer linear program modeling the problem has an integrality gap of at most $33/17$, which implies an algorithm approximating the optimal objective value with rate $33/17 + \eps$ for any $\eps > 0$ without producing a corresponding schedule.
This has been improved by Jansen and Rohwedder \cite{JansenR17SODA} to a rate of $11/6$, and in \cite{JansenR17IPCO} the same authors provide a quasi-polynomial approximation algorithm with rate $11/6 + \eps$ that also outputs a corresponding schedule.
For the special case of restricted assignment with only two distinct processing times (not counting $\infty$) an approximation algorithm due to Chakrabarty, Khanna and Li \cite{ChakrabartyKL15} with a rate slightly below $2$ is known.
Furthermore, the case in which the set of eligible machines for each job has cardinality at most $2$ has been studied under the name of graph balancing.
Ebenlendr, Krc{\'{a}}l and Sgall \cite{EbenlendrKS14} presented a $1.75$-approximation for this case.
Note that even if both of the above cases apply, there is no approximation algorithm with rate smaller than $1.5$ \cite{EbenlendrKS14} (unless P$=$NP).
However, for this special case multiple authors found a fitting $1.5$-approximation algorithm \cite{PageS16,HuangO16,ChakrabartyS16}. 
In a recent result, Jansen and Rohwedder \cite{JansenR18GrapbBalancing} showed that in the graph balancing case the optimal objective value can be approximated with a rate slightly below $1.75$.

\subparagraph*{Preliminaries.}

In the following, we deal with satisfiability problems like the classical \threesat problem where a logical formula over variables $x_1,\dots,x_n$ is given.
The formula is a conjunction of clauses, each clause is a disjunction of three literals, and a literal is either a variable or its negation.
The goal is to decide whether there is a fulfilling truth assignment, that is, an assignment of the variables to the truth values \enquote{true} and \enquote{false}, denoted by $\true$ and $\false$, respectively, such that the formula evaluates to \enquote{true}.

We consider polynomial time approximation algorithms: 
Given an instance $I$ of an optimization problem, an $\alpha$-approximation $A$ for this problem produces a solution in time $\mathrm{poly}(|I|)$, where $|I|$ denotes the input length.
For the objective function value $A(I)$ of this solution it is guaranteed that $A(I)\leq \alpha\text{\textsc{opt}}(I)$, in the case of an minimization problem, or $A(I)\geq (1/\alpha)\text{\textsc{opt}}(I)$, in the case of an maximization problem, where $\text{\textsc{opt}}(I)$ is the value of an optimal solution.
We call $\alpha$ the \emph{approximation guarantee} or \emph{rate} of the algorithm. 
In some cases a polynomial time approximation scheme (PTAS) can be achieved, that is, an $(1+\eps)$-approximation for each $\eps>0$.
If for such a family of algorithms the running time is polynomial in both $1/\eps$ and $|I|$ it is called \emph{fully polynomial} (FPTAS).

Nearly all the reductions in this work follow the same pattern:
Given an instance $I$ of the starting problem, we construct an instance $I'$ of the variant of the restricted assignment problem considered in the respective case.
For $I'$, all job sizes are integral and upper bounded by some constant $T$ such that the overall size of the jobs equals $|\machs| T $.
Obviously, if for such an instance a machine receives jobs with overall size more or \emph{less} than $T$, the makespan of the schedule is greater than $T$.
Then we show that that there exists a schedule with makespan $T$ for $I'$, if and only if $I$ is a yes-instance. 
This rules out the existence of an approximation algorithm with rate smaller than $(T+1)/T$ and a PTAS in particular.
Furthermore, for the Santa Claus version, approximation algorithms with rate smaller than $T/(T-1)$ are ruled out.

\section{Interval Restrictions}
\label{sec:interval}

The sole goal of this section is to prove Theorem \ref{thm:result_rai}, that is, the non-existence of a PTAS for \rai (given P$\neq$NP).
Our starting point for the reduction is a satisfiability problem \tailoredsat that we tailor to our needs.
We show that \tailoredsat is NP-hard via a straight forward reduction from the \oneinthreesat problem, which is well-know to be NP-complete \cite{Schaefer78} and discussed in more detail below.
Next, we provide a reduction from \tailoredsat to the classical restricted assignment problem (with arbitrary sets of eligible machines).
This reduction introduces some of the needed gadgets and ideas for the main result.
Lastly, we show how the reduction can be refined for \rai, and this is the most elaborate step.

\subparagraph{Starting Point.}

An instance of \oneinthreesat is a conjunction of clauses with 3 literals each. 
Each clause is a formula depending on 3 literals that is satisfied if and only if exactly one of its literals takes the value $\true$. 
We call such formulas 1-in-3-clauses in the following and define 2-in-3-clauses correspondingly.

An instance of the problem \tailoredsat also is a conjunction of clauses with exactly 3 literals each.
However, each of the clauses is either a 1-in-3-clause or a 2-in-3-clause and there are as many clauses of the first as of the second type.
Furthermore, we require that each literal occurs \emph{exactly twice}.
In the following, we denote a 1-in-3-clause or 2-in-3-clause with literals $z_1$, $z_2$ and $z_3$ by $(z_1,z_2,z_3)_1$ or $(z_1,z_2,z_3)_2$, respectively.

To see that \tailoredsat is NP-hard, consider an instance of \oneinthreesat with $n$ variables $x_1,\dots,x_n$ and $m$ clauses.
We now construct an equivalent \tailoredsat instance.
Let $d_i$ be the number of times the variable $x_i$ occurs in the given \oneinthreesat formula.
For each variable $x_i$, we introduce new variables $x_{i,1}, \dots, x_{i,d_i}$ and $y_{i,1}, \dots, y_{i,d_i}$ along with clauses $(x_{i,1},\neg x_{i,2},y_{i,1})_2$, \dots, $(x_{i,d_i-1},\neg x_{i,d_i},y_{i,d_i-1})_2$, $(x_{i,d_i},\neg x_{i,1},y_{i,d_i})_2$ and clauses $(y_{i,j}, \neg y_{i,j}, \neg y_{i,j})_1$ for each $j\in [d_i]$.
Note that each variable $y_{i,j}$ has to take the value $\true$ in a fulfilling assignment, due to the clause $(y_{i,j}, \neg y_{i,j}, \neg y_{i,j})_1$. 
The remaining clauses ensure, that for each $i$ the variables $x_{i,1}$, \dots $x_{i,d_i}$ have the same value in a fulfilling assignment.
Furthermore, for each of the clauses of the original problem, we introduce one 1-in-3-clause and one 2-in-3-clause.
The 1-in-3-clauses are obtained by exchanging the $j$-th occurrence of each variable $x_i$ with $x_{i,j}$.
Moreover, the 2-in-3-clauses are obtained by copying the new 1-in-3-clauses, negating all the literals and turning them into a 2-in-3-clause.
Hence, each 2-in-3-clauses evaluates to $\true$, if and only if its corresponding 1-in-3-clause does.
It is not hard to verify the correctness of the reduction.
Similar constructions are widely used, see, e.g., \cite{Tovey84} or \cite{CMYZ17}.
The remarkable aspect of the present construction lies in its symmetrical structure which helps to avoid additional dummy gadgets in the following reductions.

\subparagraph{Simple Reduction.}

In the following, we assume that an instance of \tailoredsat with $m$ 1-in-3-clauses $C_1,\dots,C_m$, $m$ 2-in-3-clauses $C_{m+1},\dots,C_{2m}$ and $n$ variables $x_1,\dots,x_n$ is given.
Note that we have $2m$ clauses with $3$ literals each, and $4n$ occurring literals in total, hence $3m = 2n$.
In addition to the ordering of the variables and clauses, we fix an ordering of the literals belonging to each clause, and an ordering of the occurrences of each variable by assigning an index $t\in [4]$ to each of them.
In particular, for each variable $x_j$, $t=1,2$ correspond to the first and second positive and $t=3,4$ to the first and second negative occurrence of $x_j$.
Furthermore, let $\kappa: [n]\times [4]\rightarrow [2m]\times [3]$ be the bijection defined as follows: $\kappa(j,t) = (i,s)$ implies that the $t$-th occurrence of $x_j$ is positioned in clause $C_i$ on position $s$.

We now define the restricted assignment instance.
For some of the machines, we introduce \emph{private loads} which is a synonym for jobs of the corresponding size that have to be scheduled on the respective machine because its the only eligible one.
The sizes and sets of eligible machines of the introduced jobs are presented in \cref{table:simple_size_eligible} and the target makespan is given by $T = 322$.
\begin{itemize}

\item For each clause $C_i$, there are three \emph{clause machines} $\clausemach_{i,s}$ with $s\in[3]$ corresponding to its three literals, as well as three \emph{clause jobs} $\clausejob_{i,s'}^{\circ_{s'}}$ with $s'\in[3]$ and $\circ_{s'}\in\set{\true,\false}$.
We have $\circ_1 = \true$ and $\circ_3 = \false$, as well as $\circ_2 = \false$ if $C_i$ is a 1-in-3 clause, and $\circ_2=\true$ otherwise.
Furthermore, each clause machine has a private load of $111$.

\item For each variable $x_j$, there are two \emph{truth assignment machines} $\truthmach_{j,q}$ with $q\in [2]$ corresponding to the positive ($q=1$) and negative ($q=2$) literal of $x_j$, as well as 2 \emph{truth assignment jobs} $\truthjob_{j}^{\circ}$ with $\circ\in\set{\true,\false}$. 

\item For each variable $x_j$, there are eight \emph{variable jobs} $\variablejob_{j,t}^{\circ}$ with $t\in[4]$ and $\circ\in\set{\true,\false}$ corresponding to the two occurrences of the positive ($t\in\set{1,2}$) and negative ($t\in\set{3,4}$) literal of $x_j$.

\end{itemize}
\begin{table}
\centering
\caption{The sizes and sets of eligible machines of the jobs in the simple reduction. The entry for $\clausemach_{i,s}$ marks the private load of the machine. The target makespan is given by $T = 322$.}
\begin{tabular}{lll}
\toprule
Job & Size & Eligible Machines \\ \midrule
$\clausemach_{i,s}$ & $111$ & $\clausemach_{i,s}$\\
$\clausejob_{i,s'}^{\true}$ & $100$ & $\clausemach_{i,1}$, $\clausemach_{i,2}$, $\clausemach_{i,3}$\\
$\clausejob_{i,s'}^{\false}$ & $101$ & $\clausemach_{i,1}$, $\clausemach_{i,2}$, $\clausemach_{i,3}$\\
$\truthjob_{j}^{\true}$ & $100$ & $\truthmach_{j,1}$, $\truthmach_{j,2}$ \\
$\truthjob_{j}^{\false}$ & $102$ & $\truthmach_{j,1}$, $\truthmach_{j,2}$ \\
$\variablejob_{j,t}^{\true}$ & $111$ & $\truthmach_{j,\ceil{t/2}}$, $\clausemach_{\kappa(j,t)}$ \\
$\variablejob_{j,t}^{\false}$ & $110$ & $\truthmach_{j,\ceil{t/2}}$, $\clausemach_{\kappa(j,t)}$ \\
\bottomrule
\end{tabular}
\label{table:simple_size_eligible}
\end{table}
First note:
\begin{claim}
The overall size of all the jobs is exactly $|\machs|T$.
\end{claim}
\begin{claimproof}
We have $6m + 2n = 6n$ machines (since $3m = 2n$), and---taking into account that we have as many 1-in-3 as 2-in-3 clauses---the overall job size equals: \[6m\cdot 111 + m(3\cdot 100 + 3\cdot 101) + n(100 + 102 + 4\cdot 111 + 4\cdot 110) = 1932n = 322 \cdot 6n = |\machs|T\]
\end{claimproof}
We will show that there is a fulfilling truth assignment for the \tailoredsat instance if and only if there is a schedule in which each machine receives jobs with load exactly $T$.

For any job $\mathtt{Job}^\circ$ with $\circ\in\set{\true,\false}$, we refer to $\circ$ as its \emph{truth configuration} and say that $\mathtt{Job}^\circ$ has $\circ$-configuration.
The rationale of the reduction is as follows:
Each clause machine $\clausemach_{i,s}$ should receive exactly one variable job corresponding to the literal placed in position $s$ in the clause.
The truth configuration of this variable job should correspond to the truth value the variable contributes to the clause.
To ensure that the jobs $\variablejob_{j,t}^{\true}$ belonging to variable $x_j$ contribute consistent truth values, the truth assignment jobs and machines are introduced.
In the following, we sometimes talk about the truth assignment gadget and thus refer to these jobs and machines. 
Similarly, the clause machines and jobs are sometimes called the clause gadget.
In the appendix, we provide an example \tailoredsat instance and the corresponding restricted assignment instance produced in the reduction.

Next, we present a sequence of easy claims concerning the properties of a schedule for the above instance with makespan $T$.
\begin{claim}\label{claim:simple_3jobs}
Each machine receives exactly 3 jobs (including private loads).
\end{claim}
\begin{claimproof}
Since the overall size of the jobs is $|\machs|T$, we know that each machine has to receive jobs with overall size $T = 322$. Each job or private load has a size of at least $100$ and at most~$111$.
\end{claimproof}
Since each digit of each occurring size is upper bounded by $2$, the above claim implies that there can be no carryover when adding up job sizes of jobs scheduled on each machine.
Hence the digits of the numbers involved can be considered independently, e.g., there can be at most two jobs with a $1$ in the third (or second) digit of its size scheduled on any machine.
This together with the given job restrictions already implies:
\begin{claim}\label{claim:simple_TMach_CMach_Assignment}
Each truth assignment machine receives exactly one truth assignment and two variable jobs; and each clause machine receives exactly one clause and one variable job.
\end{claim}
\begin{claim}\label{claim:simple_TMach_CMach_Truth_Conf}
The jobs scheduled on a truth assignment or clause machine all have the same truth configuration (excluding private loads).
\end{claim}
\begin{claim}\label{claim:simple_TJobs_Truth_Assignment}
Let $j\in[n]$. The truth configuration of any job scheduled on $\truthmach_{j,1}$ is distinct from the truth configuration of any job scheduled on $\truthmach_{j,2}$.
\end{claim}
\begin{table}
\centering
\caption{Each set indicates one of the possible job assignments for each machine in a schedule with makespan $T$.}
\begin{tabular}{ll}
\toprule
Machine & Possible Schedules \\ \midrule
$\truthmach_{j,1}$ & $\set{\truthjob_{j}^\true, \variablejob_{j,1}^\true, \variablejob_{j,2}^\true}$, $\set{\truthjob_{j}^\false, \variablejob_{j,1}^\false, \variablejob_{j,2}^\false}$ \\
$\truthmach_{j,2}$ & $\set{\truthjob_{j}^\true, \variablejob_{j,3}^\true, \variablejob_{j,4}^\true}$, $\set{\truthjob_{j}^\false, \variablejob_{j,3}^\false, \variablejob_{j,4}^\false}$ \\
$\clausemach_{i,s}$ (1-in-3-clause) & $\set{\variablejob_{\kappa^{-1}(i,s)}^{\true}, \clausejob_{i,1}^\true}$, $\set{\variablejob_{\kappa^{-1}(i,s)}^{\false}, \clausejob_{i,2}^\false}$, $\set{\variablejob_{\kappa^{-1}(i,s)}^{\false}, \clausejob_{i,3}^\false}$\\
$\clausemach_{i,s}$ (2-in-3-clause) & $\set{\variablejob_{\kappa^{-1}(i,s)}^{\true}, \clausejob_{i,1}^\true}$, $\set{\variablejob_{\kappa^{-1}(i,s)}^{\true}, \clausejob_{i,2}^\true}$, $\set{\variablejob_{\kappa^{-1}(i,s)}^{\false}, \clausejob_{i,3}^\false}$\\
\bottomrule
\end{tabular}
\label{table:simple_possible_jobs}
\end{table}

\begin{figure}
\centering
\begin{tikzpicture}
\pgfmathsetmacro{\bw}{1.8}
\pgfmathsetmacro{\bh}{2.6}
\pgfmathsetmacro{\tth}{.7}
\pgfmathsetmacro{\tfh}{1.1}
\pgfmathsetmacro{\vth}{0.8}
\pgfmathsetmacro{\vfh}{0.6}

\draw[thick] (0, 0) -- (0,\bh); 
\draw[thick] (0,0) -- (\bw,0) node [midway,yshift=-13pt] {$\truthmach_{j,1}$};
\draw[thick] (\bw, 0) -- (\bw,\bh); 
\draw[thick] (\bw,0) -- (2*\bw,0) node [midway,yshift=-13pt] {$\truthmach_{j,2}$};
\draw[thick] (2*\bw, 0) -- (2*\bw,\bh); 

\draw[thick] (0,0) rectangle (\bw, \tth) node [midway] {$\truthjob_{j}^\true$};
\draw[thick] (0,\tth) rectangle (\bw, \tth + \vth) node [midway] {$\variablejob_{j,1}^\true$};
\draw[thick] (0,\tth + \vth) rectangle (\bw, \tth + 2*\vth) node [midway] {$\variablejob_{j,2}^\true$};

\draw[thick] (\bw,0) rectangle (2*\bw, \tfh) node [midway] {$\truthjob_{j}^\false$};
\draw[thick] (\bw,\tfh) rectangle (2*\bw, \tfh + \vfh) node [midway] {$\variablejob_{j,3}^\false$};
\draw[thick] (\bw,\tfh + \vfh) rectangle (2*\bw, \tfh + 2*\vfh) node [midway] {$\variablejob_{j,4}^\false$};

\draw[thick, ->, >=latex] (\bw/2, \bh - 0.2) to [out=90,in=190] (\bw/2 + 0.5, \bh - 0.2 + 1.8) node [right] {$\variablejob_{j,1}^{\false}$}; 
\draw[thick, ->, >=latex] (\bw/2, \bh - 0.2) to [out=90,in=190] (\bw/2 + 0.7, \bh - 0.2 + 1.3) node [right] {$\variablejob_{j,2}^{\false}$}; 
\draw[thick, ->, >=latex] (\bw + \bw/2, \bh - 0.2) to [out=90,in=190] (\bw + \bw/2 + 1.0, \bh - 0.2 + 1.8) node [right] {$\variablejob_{j,3}^{\true}$}; 
\draw[thick, ->, >=latex] (\bw + \bw/2, \bh - 0.2) to [out=90,in=190] (\bw + \bw/2 + 1.2, \bh - 0.2 + 1.3) node [right] {$\variablejob_{j,4}^{\true}$}; 

\begin{scope}[xshift=6cm]

\draw[thick] (0, 0) -- (0,\bh); 
\draw[thick] (0,0) -- (\bw,0) node [midway,yshift=-13pt] {$\truthmach_{j,1}$};
\draw[thick] (\bw, 0) -- (\bw,\bh); 
\draw[thick] (\bw,0) -- (2*\bw,0) node [midway,yshift=-13pt] {$\truthmach_{j,2}$};
\draw[thick] (2*\bw, 0) -- (2*\bw,\bh); 

\draw[thick] (0,0) rectangle (\bw, \tfh) node [midway] {$\truthjob_{j}^\false$};
\draw[thick] (0,\tfh) rectangle (\bw, \tfh + \vfh) node [midway] {$\variablejob_{j,1}^\false$};
\draw[thick] (0,\tfh + \vfh) rectangle (\bw, \tfh + 2*\vfh) node [midway] {$\variablejob_{j,2}^\false$};

\draw[thick] (\bw,0) rectangle (2*\bw, \tth) node [midway] {$\truthjob_{j}^\true$};
\draw[thick] (\bw,\tth) rectangle (2*\bw, \tth + \vth) node [midway] {$\variablejob_{j,3}^\true$};
\draw[thick] (\bw,\tth + \vth) rectangle (2*\bw, \tth + 2*\vth) node [midway] {$\variablejob_{j,4}^\true$};

\draw[thick, ->, >=latex] (\bw/2, \bh - 0.2) to [out=90,in=190] (\bw/2 + 0.5, \bh - 0.2 + 1.8) node [right] {$\variablejob_{j,1}^{\true}$}; 
\draw[thick, ->, >=latex] (\bw/2, \bh - 0.2) to [out=90,in=190] (\bw/2 + 0.7, \bh - 0.2 + 1.3) node [right] {$\variablejob_{j,2}^{\true}$}; 
\draw[thick, ->, >=latex] (\bw + \bw/2, \bh - 0.2) to [out=90,in=190] (\bw + \bw/2 + 1.0, \bh - 0.2 + 1.8) node [right] {$\variablejob_{j,3}^{\false}$}; 
\draw[thick, ->, >=latex] (\bw + \bw/2, \bh - 0.2) to [out=90,in=190] (\bw + \bw/2 + 1.2, \bh - 0.2 + 1.3) node [right] {$\variablejob_{j,4}^{\false}$}; 

\end{scope}
\end{tikzpicture}

\caption{The truth assignment gadget: There are two possible schedules of the truth assignment machines $\truthmach_{j,1}$ and $\truthmach_{j,2}$ that already determine the schedule of the variable jobs.}
\label{fig:truth_assignment_gadget}
\end{figure}
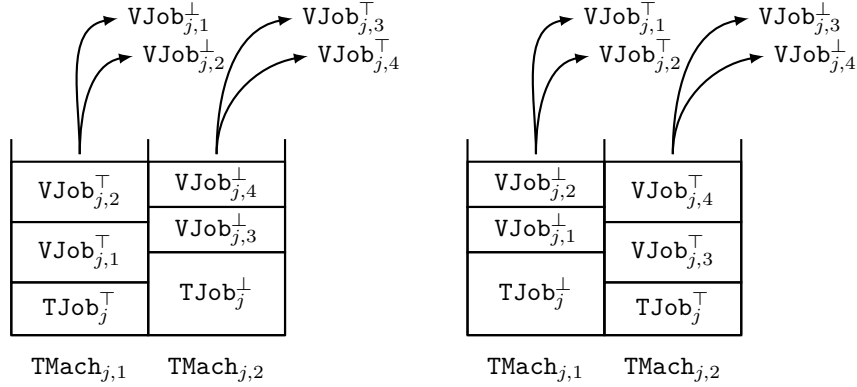
The resulting possible schedules for each machine are summed up in Table~\ref{table:simple_possible_jobs}, and Figure~\ref{fig:truth_assignment_gadget} depicts the resulting two possible schedules for each pair of truth assignment machines.
Lastly, we have:
\begin{claim}\label{claim:simple_CMach_triplet_truth_configuration}
For each $i\in [2m]$, the three clause machines corresponding to $i$ receive exactly one variable job with $\true$-configuration if $C_i$ is a 1-in-3-clause and exactly two such jobs if $C_i$ is a 2-in-3-clause.
\end{claim}
\begin{claimproof}
The overall load on each triplet of clause machines has to be $3T = 966$ and the private loads and clause jobs that have to be scheduled on the triplet have summed up load $635$, in case of a 1-in-3-clause, and $634$, in case of a 2-in-3-clause.
The only other jobs eligible on the clause machines are variable jobs with size $111$ in $\true$-configuration and $110$ otherwise.
This implies the claim.
\end{claimproof}
Using the above claims, we can easily show:
\begin{proposition}
There is a fulfilling truth assignment for the given \tailoredsat instance if and only if there is a schedule with makespan $T$ for the constructed restricted assignment instance. 
\end{proposition}
\begin{proof}
Let there be a schedule with makespan $T$ for the constructed instance.
For each variable $x_j$ and occurrence $t\in[4]$, let $\variablejob_{j,1}^{\circ_{j,t}}$ be the variable job scheduled on $\clausemach_{\kappa(j,t)}$ (see Table \ref{table:simple_possible_jobs}).
We choose the truth value of $x_j$ to be $\circ_{j,1}$.
The variable $x_j$ occurs exactly four times in the formulas, namely as a positive literal on the positions $\kappa(j,1)$ and $\kappa(j,2)$ and as a negative literal at position $\kappa(j,3)$ and $\kappa(j,4)$.
Because of the above observations (see Figure \ref{fig:truth_assignment_gadget}), we know that $\circ_{j,2} = \circ_{j,1}$ and $\circ_{j,3} = \circ_{j,4} \neq \circ_{j,1}$.
Hence, for each variable $x_j$ and occurrence $t\in[4]$, the truth configuration $\variablejob_{j,1}^{\circ_{j,t}}$ corresponds exactly to the truth value $x_j$ contributes to the clause given by $\kappa(j,t)$.
Lastly, for each clause $C_i$, there are exactly three variable jobs scheduled on the corresponding clause machines, and exactly one or two of these has $\true$-configuration, if $C_i$ is a 1-in-3-clause or 2-in-3-clause respectively (Claim \ref{claim:simple_CMach_triplet_truth_configuration}).
Hence, $C_i$ is fulfilled.

Next, we consider the case that a fulfilling truth assignment is given.
For each variable $x_j$, let $\triangleleft_j$ be the corresponding truth value and $\triangleright_j$ its negation.
We set $\circ_{j,t} = \triangleleft_j$ for $t\in\set{1,2}$ and $\circ_{j,t} = \triangleright_j$ for $t\in\set{3,4}$ and assign $\variablejob_{j,1}^{\circ_{j,t}}$ to $\clausemach_{\kappa(j,t)}$.
All the other jobs are assigned as indicated by Table \ref{table:simple_possible_jobs} and Figure \ref{fig:truth_assignment_gadget}.
It is easy to verify, that all jobs are assigned and each machine has a load of $T$.
\end{proof}
The basic approach of using some kind of truth assignment and clause gadget for reductions in the context of restricted assignment and unrelated scheduling has been used before, see, e.g., \cite{CMYZ17,EbenlendrKS14}.

\subparagraph{Refined Reduction.}

When trying to adapt the above reduction to the more restricted problem of \rai, we obviously have less latitude when defining the restrictions.
To deal with this, we introduce additional gadgets and encode much more information into the job sizes.
The idea of the reduction can be described as follows.
We arrange the truth assignment gadgets on the left and the clause gadgets on the right.
Consider the case that a truth assignment decision is made in the left most truth assignment gadget.
Information about this decision---called signal in the following---has to be passed on to the proper clause gadgets passing multiple other truth assignment and clause gadgets on the way.
This signal in the simple reduction simply corresponds to a variable job that is to be scheduled on its corresponding clause machine, and in order to prevent interaction with other gadgets, we could encode information about the corresponding variable into the size of the variable job.
However, this would lead to a super constant number of job sizes.
To avoid this, we introduce a new gadget called the bridge and highway gadget.
Very roughly speaking, the signal is passed on to the \emph{highway} via \emph{gateways}; the highway passes each following truth assignment gadget using \emph{bridges} and carries the signal to the proper clauses. 
Next, we give a detailed description and analysis of the refined reduction.

We adopt all the machines and jobs introduced in the simple reduction, but change the sizes and sets of eligible machines and introduce additional jobs and machines as well as private loads for \emph{every} machine.
We introduce the following jobs and machines:
\begin{itemize}
	
\item For each $j\in [n]$ and $t\in[4]$, we introduce one \emph{gateway machine} $\gatemach_{j,t}$.

\item For each $j\in [n]$, $t\in[4]$ and $j'\in \set{j+1,\dots n}$, we introduce two \emph{bridge machines} $\bridgein_{j,t,j'}$ and $\bridgeout_{j,t,j'}$.
Furthermore, we introduce two \emph{bridge jobs} $\bridgejob_{j,t,j'}^\true$ and $\bridgejob_{j,t,j'}^\false$.

\item For each $j\in [n]$, $t\in[4]$ and $j'\in \set{j,\dots n}$, we introduce two \emph{highway jobs} $\highwayjob_{j,t,j'}^\true$ and $\highwayjob_{j,t,j'}^\false$.

\end{itemize}
In order to define the intervals of eligible machines, we first need a total order of the machines.
We partition the machines into blocks, define an internal order for each block, and then define an order of the blocks.
Remember that $\kappa: [n]\times [4]\rightarrow [2m]\times [3]$ is a bijection indicating the positions of the occurrences of variables in the clauses.
In particular, $\kappa(j,1) = (i,s)$ indicates that the first positive occurrence of variable $x_j$ is in clause $C_i$ on position $s$, and $\kappa(j,2)$, $\kappa(j,3)$, and $\kappa(j,4)$ indicate analogue information for the second positve, first negative, and second negative occurrence of $x_j$.
\begin{itemize}
	
\item For each $j\in [n]$, we have a truth assignment block $\tblock_j$ containing the truth assignment machines $\truthmach_{j,1}$ and $\truthmach_{j,2}$ in this order.

\item For each $i\in[2m]$, we have a clause block $\cblock_i$ containing the clause machines $\clausemach_{i,s}$ for each $s\in [3]$ and ordered increasingly by $s$.

\item For each $j\in [n]$, we have a successor block $\sblock_j$ containing the gateway machines $\gatemach_{j,t}$ for each $t\in [4]$ and the bridge machines $\bridgeout_{j',t,j}$ for each $t\in [4]$ and $j' < j$.
For each machine, we define an index, namely $\kappa(j,t)$ for $\gatemach_{j,t}$ and $\kappa(j',t)$ for $\bridgeout_{j',t,j}$, and order the machines by the \emph{decreasing} lexicographical ordering of their indices.
For example, if $\bridgeout_{j_1,t_1,j},\bridgeout_{j_2,t_2,j},\gatemach_{j,t_3} \in\sblock_j$ and $\kappa(j_1, t_1) = (1,2)$, $\kappa(j_2, t_2) = (1,1)$ and $\kappa(j, t_3) = (2,3)$, then $\gatemach_{j,t_3}$ precedes $\bridgeout_{j_1,t_1,j}$ which in turn precedes $\bridgeout_{j_2,t_2,j}$.

\item  For each $j\in [n]$ with $j>1$, we have a predecessor block $\pblock_j$ containing the bridge machines $\bridgein_{j',t,j}$ for each $t\in [4]$ and $j' < j$.
Machine $\bridgein_{j',t,j}$ has index $\kappa(j',t)$ and the machines are ordered by the \emph{increasing} lexicographical ordering of their indices.
\end{itemize}
The blocks are ordered as follows:
\[(\tblock_1,\sblock_1,\pblock_2,\tblock_2,\sblock_2,\dots,\pblock_n,\tblock_n,\sblock_n,\cblock_1,\dots,\cblock_{2m})\]
The sets of eligible machines are specified in Table \ref{table:eligible_machs} and the job sizes in Table \ref{table:job_sizes}\footnote{Note that we have prioritized comprehensibility over small sizes. For instance, it is not hard to see that the columns in \cref{table:job_sizes} corresponding to the highway and clause jobs could be deleted and the reduction would still work.}.
In the appendix, we provide an example instance together with a figure (Figure \ref{fig:interval_example}) visualizing the ordering of the machines and the eligibility constraints.
Furthermore, Figure \ref{fig:interval} gives some intuition on the overall structure.
We have:
\begin{claim}\label{claim:refined_overall_load}
The overall size of the jobs is exactly $|\machs|T$.
\end{claim}
\begin{claimproof}
This can be verified by basic arithmetic using Table \ref{table:job_sizes}.
For simplicity, this can also be done digit by digit.
We look at the last digit as an example.
Note that we have $4n^2 + 6n$ machines and the last digit of the makespan is $2$.
Summing up the last digits of the job sizes, on the other hand, yields $2n + 2n + n + n + n + n + 2n(n-1) + 2n(n+1) + 4n + n + n + n + n + 2n(n-1) + 2n(n-1) = 8n^2 + 12n$.
\end{claimproof}

\begin{table}
\centering
\caption{The sets of eligible machines for each job or job type, defined by the first and last eligible machine in the ordering. Note that in case of the highway jobs all four combinations of first and last machine are possible.}
\begin{tabularx}{1\textwidth}{Xp{4cm}p{4cm}}
\toprule
Job & First machine & Last machine \\\midrule
Clause job $\clausejob_{i,s}^{\circ_s}$ & $\clausemach_{i,1}$ & $\clausemach_{i,3}$ \\
Truth assignment job $\truthjob_{j}^{\circ}$ & $\truthmach_{j,1}$ & $\truthmach_{j,2}$\\
Variable job $\variablejob_{j,t}^{\circ}$ & $\truthmach_{j,\ceil{t/2}}$ & $\gatemach_{j,t}$\\
Bridge job $\bridgejob_{j,t,j'}^{\circ}$ & $\bridgein_{j,t,j'}$ & $\bridgeout_{j,t,j'}$ \\
Highway job $\highwayjob_{j,t,j'}^{\circ}$ & $\bridgeout_{j,t,j'}$, if $j'>j$, $\gatemach_{j,t}$, if $j' = j$ & $\bridgein_{j,t,j'+1}$ if $j'<n$, $\clausemach_{\kappa(j,t)}$, if $j' = n $  \\\bottomrule
\end{tabularx}
\label{table:eligible_machs}

\vspace*{\floatsep}

\caption{Table of job and machine types with job sizes and private loads and the makespan.
The second column states the number of jobs and machines of the respective types. 
Each horizontal sequence of numbers following the second column indicates the size of the respective job or private load. 
Each of the corresponding columns serves a function in the reduction: the first bounds the number of jobs on each machines; the following eight implement restrictions for the bridge, highway, clause, truth assignment and variable jobs; and the last encodes truth values.}
\begin{tabularx}{0.74\textwidth}{X|c||l|l|l|l|l|llll|l}
 & \# & & $\mathtt{B}$ & $\mathtt{H}$ & $\mathtt{C}$ & $\mathtt{T}$ & $\mathtt{V}$ & $\mathtt{V}$ & $\mathtt{V}$ & $\mathtt{V}$ & \\
\hline
$\clausejob_{i,s}^{\true}$   	& $3m = 2n$ & 1 & 0 & 0 & 1 & 0 & 0 & 0 & 0 & 0 & 0\\
$\clausejob_{i,s}^{\false}$  	& $3m = 2n$ & 1 & 0 & 0 & 1 & 0 & 0 & 0 & 0 & 0 & 1\\
\hline
$\truthjob_{j}^{\true}$      	& $n $ & 1 & 0 & 0 & 0 & 1 & 0 & 0 & 0 & 0 & 2\\
$\truthjob_{j}^{\false}$     	& $n $ & 1 & 0 & 0 & 0 & 1 & 0 & 0 & 0 & 0 & 0\\
\hline
$\variablejob_{j,1}^{\true}$ 	& $n $ & 1 & 0 & 0 & 0 & 0 & 1 & 0 & 0 & 0 & 0\\
$\variablejob_{j,2}^{\true}$ 	& $n $ & 1 & 0 & 0 & 0 & 0 & 0 & 1 & 0 & 0 & 0\\
$\variablejob_{j,3}^{\true}$ 	& $n $ & 1 & 0 & 0 & 0 & 0 & 0 & 0 & 1 & 0 & 0\\
$\variablejob_{j,4}^{\true}$  	& $n $ & 1 & 0 & 0 & 0 & 0 & 0 & 0 & 0 & 1 & 0\\
$\variablejob_{j,1}^{\false}$ 	& $n $ & 1 & 0 & 0 & 0 & 0 & 1 & 0 & 0 & 0 & 1\\
$\variablejob_{j,2}^{\false}$ 	& $n $ & 1 & 0 & 0 & 0 & 0 & 0 & 1 & 0 & 0 & 1\\
$\variablejob_{j,3}^{\false}$ 	& $n $ & 1 & 0 & 0 & 0 & 0 & 0 & 0 & 1 & 0 & 1\\
$\variablejob_{j,4}^{\false}$ 	& $n $ & 1 & 0 & 0 & 0 & 0 & 0 & 0 & 0 & 1 & 1\\
\hline
$\bridgejob_{j,t,j'}^{\true}$	& $2n(n-1) $ & 1 & 1 & 0 & 0 & 0 & 0 & 0 & 0 & 0 & 0\\
$\bridgejob_{j,t,j'}^{\false}$	& $2n(n-1) $ & 1 & 1 & 0 & 0 & 0 & 0 & 0 & 0 & 0 & 1\\
\hline
$\highwayjob_{j,t,j'}^{\true}$	& $2n(n+1) $ & 1 & 0 & 1 & 0 & 0 & 0 & 0 & 0 & 0 & 1\\
$\highwayjob_{j,t,j'}^{\false}$	& $2n(n+1) $ & 1 & 0 & 1 & 0 & 0 & 0 & 0 & 0 & 0 & 0\\
\hline
$\clausemach_{i,s}$ 			& $6m = 4n $ & 1 & 1 & 0 & 0 & 1 & 1 & 1 & 1 & 1 & 1\\
\hline
$\truthmach_{j,1}$ 				& $n $ & 0 & 1 & 1 & 1 & 0 & 0 & 0 & 1 & 1 & 0\\
$\truthmach_{j,2}$ 				& $n $ & 0 & 1 & 1 & 1 & 0 & 1 & 1 & 0 & 0 & 0\\
\hline
$\gatemach_{j,1}$ 				& $n $ & 1 & 1 & 0 & 1 & 1 & 0 & 1 & 1 & 1 & 1\\
$\gatemach_{j,2}$ 				& $n $ & 1 & 1 & 0 & 1 & 1 & 1 & 0 & 1 & 1 & 1\\
$\gatemach_{j,3}$ 				& $n $ & 1 & 1 & 0 & 1 & 1 & 1 & 1 & 0 & 1 & 1\\
$\gatemach_{j,4}$ 				& $n $ & 1 & 1 & 0 & 1 & 1 & 1 & 1 & 1 & 0 & 1\\
\hline
$\bridgein_{j,t,j'}$ 			& $2n(n-1) $ & 1 & 0 & 0 & 1 & 1 & 1 & 1 & 1 & 1 & 1\\
$\bridgeout_{j,t,j'}$ 			& $2n(n-1) $ & 1 & 0 & 0 & 1 & 1 & 1 & 1 & 1 & 1 & 1\\
\hline
Makespan $T$ 					& & 3 & 1 & 1 & 1 & 1 & 1 & 1 & 1 & 1 & 2\\
\end{tabularx}
\label{table:job_sizes}
\end{table}

Like for the simple reduction, we proof a sequence of easy claims concerning the properties of a schedule for the constructed instance with makespan $T$.
\begin{claim}\label{claim:refined_4_or_3_jobs_per_machine}
Each machine receives exactly 4 jobs if it is a truth assignment machine and exactly 3 jobs otherwise (including private loads).
\end{claim}
\begin{claimproof}
This follows directly from \cref{claim:refined_overall_load} and the job sizes defined in Table \ref{table:job_sizes}.
\end{claimproof}
Since each machine receives at most $4$ jobs and each digit in the job sizes is bounded by $2$, we may consider each digit of the involved numbers independently, e.g., if two jobs and the makespan have a $1$ at the $\ell$-th digit, we already know that these jobs cannot be scheduled on the same machine.
This already implies a series of claims:
\begin{claim}\label{claim:refined_TJob_restriction_and_distribution}
The jobs $\truthjob_{j}^{\true}$ and $\truthjob_{j}^{\false}$ can exclusively be scheduled on $\truthmach_{j,1}$ and $\truthmach_{j,2}$, for each $j\in [n]$, and each of the two machines receives exactly one of the two jobs.
\end{claim}
\begin{claim}\label{claim:refined_VJob_restriction_and_distribution}
The jobs $\variablejob_{j,t}^{\true}$ and $\variablejob_{j,t}^{\false}$ can exclusively be scheduled on $\truthmach_{j,\ceil{t/2}}$ and $\gatemach_{j,t}$, for each $j\in [n]$ and $t\in [4]$, and each of the two machines receives exactly one of the two jobs.
\end{claim}
\begin{claim}\label{claim:refined_BJob_restriction_and_distribution}
Bridge jobs can exclusively be scheduled on bridge machines and each bridge machine receives exactly one bridge job.
\end{claim}
\begin{claim}\label{claim:refined_HJob_restriction_and_distribution}
Highway jobs can exclusively be scheduled on bridge, gateway and clause machines and each such machine receives exactly one highway job.
\end{claim}
\begin{claim}\label{claim:refined_CJob_distribution}
Each clause machine $\clausemach_{i,s}$ receives exactly one of the corresponding clause jobs $\clausejob_{i,s'}^{\circ_{s'}}$ with $s'\in [3]$.
\end{claim}
At this point, we already know that variable (and truth assignment) jobs can exclusively be scheduled on the first or last machine of their respective interval of eligible machines.
The next step is to show that the same holds for highway and bridge jobs.
To do so, the ordering of the bridge and highway machines is of critical importance.
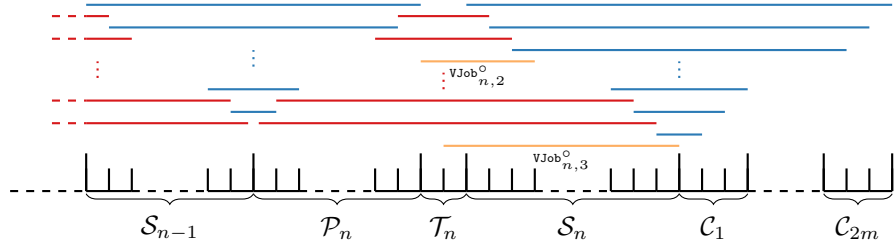
\begin{figure}
\centering

\begin{tikzpicture}
\pgfmathsetmacro{\mw}{0.3}
\pgfmathsetmacro{\gw}{1}
\pgfmathsetmacro{\bbh}{0.5}
\pgfmathsetmacro{\bsh}{0.3}
\pgfmathsetmacro{\ih}{0.3}
\pgfmathsetmacro{\ph}{0.22}
\pgfmathsetmacro{\sh}{0.15}

\foreach \x/\y in {0/0, 1/2, 2/6, 3/13, 4/19}{
    \draw[thick,dashed] (\x*\gw + \y*\mw,0) -- (\x*\gw + \gw + \y*\mw,0);
    \draw[thick] (\x*\gw + \gw + \y*\mw,\bsh) -- (\x*\gw + \gw + \y*\mw,0);
}

\foreach \x/\y in {1/0, 2/4, 3/8, 3/10, 4/16, 4/19, 5/19, 5/22}{
    \draw[thick] (\x*\gw + \y*\mw,\bbh) -- (\x*\gw + \y*\mw,0);
}

\foreach \x/\y in {1/0, 1/1, 2/2, 2/3, 2/4, 2/5, 3/6, 3/7, 3/8, 3/9, 3/10, 3/11, 3/12, 4/13, 4/14, 4/15, 4/16, 4/17, 4/18, 5/19, 5/20, 5/21}{
    \draw[thick] (\x*\gw + \y*\mw,0) -- (\x*\gw + \y*\mw + \mw,0);
    \draw[thick] (\x*\gw + \y*\mw + \mw,\bsh) -- (\x*\gw + \y*\mw + \mw,0);
}

\draw[decorate,decoration={brace,amplitude=4pt,raise=1pt,mirror},yshift=-0pt] (1*\gw + 0*\mw,0) -- (2*\gw + 4*\mw,0) node [midway,yshift=-14pt]{$\sblock_{n-1}$};
\draw[decorate,decoration={brace,amplitude=4pt,raise=1pt,mirror},yshift=-0pt] (2*\gw + 4*\mw,0) -- (3*\gw + 8*\mw,0) node [midway,yshift=-14pt]{$\pblock_{n}$};
\draw[decorate,decoration={brace,amplitude=4pt,raise=1pt,mirror},yshift=-0pt] (3*\gw + 8*\mw,0) -- (3*\gw + 10*\mw,0) node [midway,yshift=-14pt]{$\tblock_{n}$};
\draw[decorate,decoration={brace,amplitude=4pt,raise=1pt,mirror},yshift=-0pt] (3*\gw + 10*\mw,0) -- (4*\gw + 16*\mw,0) node [midway,yshift=-14pt]{$\sblock_{n}$};
\draw[decorate,decoration={brace,amplitude=4pt,raise=1pt,mirror},yshift=-0pt] (4*\gw + 16*\mw,0) -- (4*\gw + 19*\mw,0) node [midway,yshift=-14pt]{$\cblock_{1}$};
\draw[decorate,decoration={brace,amplitude=4pt,raise=1pt,mirror},yshift=-0pt] (5*\gw + 19*\mw,0) -- (5*\gw + 22*\mw,0) node [midway,yshift=-14pt]{$\cblock_{2m}$};

\draw[thick, color = color3] (3*\gw + 9*\mw,\bsh + \ih) -- (4*\gw + 16*\mw,\bsh + \ih) node [midway, yshift = - 0.2cm, text = black] {\tiny $\variablejob_{n,3}^{\circ}$};
\draw[thick, color = color1] (2*\gw + 4*\mw + 0.07,\bsh + 2*\ih) -- (4*\gw + 15*\mw,\bsh + 2* \ih);
\draw[thick, color = color1] (2*\gw + 5*\mw,\bsh + 3*\ih) -- (4*\gw + 14*\mw,\bsh + 3* \ih);

\draw[thick, color = color1 , dotted] (3*\gw + 9*\mw,\bsh + 3*\ih + \sh) -- (3*\gw + 9*\mw,\bsh + 3* \ih + \sh + \ph);

\draw[thick, color = color3] (3*\gw + 8*\mw,\bsh + 3*\ih + 2*\sh + \ph) -- (3*\gw + 13*\mw,\bsh + 3*\ih + 2*\sh + \ph) node [midway, yshift = -0.2cm, text = black] {\tiny $\variablejob_{n,2}^{\circ}$};
\draw[thick, color = color1] (3*\gw + 6*\mw,\bsh + 4*\ih + 2*\sh + \ph) -- (3*\gw + 12*\mw,\bsh + 4*\ih + 2*\sh + \ph);
\draw[thick, color = color1] (3*\gw + 7*\mw,\bsh + 5*\ih + 2*\sh + \ph) -- (3*\gw + 11*\mw,\bsh + 5*\ih + 2*\sh + \ph);

\draw[thick, color = color2] (4*\gw + 15*\mw,\bsh + 2* \ih - \sh) -- (4*\gw + 17*\mw,\bsh + 2* \ih - \sh);
\draw[thick, color = color2] (4*\gw + 14*\mw,\bsh + 3* \ih - \sh) -- (4*\gw + 18*\mw,\bsh + 3* \ih - \sh);
\draw[thick, color = color2] (4*\gw + 13*\mw,\bsh + 4* \ih - \sh) -- (4*\gw + 19*\mw,\bsh + 4* \ih - \sh);

\draw[thick, color = color2 , dotted] (4*\gw + 16*\mw,\bsh + 4*\ih + 0*\sh ) -- (4*\gw + 16*\mw,\bsh + 4* \ih + 0*\sh + \ph);

\draw[thick, color = color2] (3*\gw + 12*\mw,\bsh + 4*\ih + 1*\sh + \ph) -- (5*\gw + 20*\mw,\bsh + 4*\ih + 1*\sh + \ph);
\draw[thick, color = color2] (3*\gw + 11*\mw,\bsh + 5*\ih + 1*\sh + \ph) -- (5*\gw + 21*\mw,\bsh + 5*\ih + 1*\sh + \ph);
\draw[thick, color = color2] (3*\gw + 10*\mw,\bsh + 6*\ih + 1*\sh + \ph) -- (5*\gw + 22*\mw,\bsh + 6*\ih + 1*\sh + \ph);

\draw[thick, color = color2] (2*\gw + 3*\mw,\bsh + 3* \ih - \sh) -- (2*\gw + 5*\mw,\bsh + 3* \ih - \sh);
\draw[thick, color = color2] (2*\gw + 2*\mw,\bsh + 4* \ih - \sh) -- (2*\gw + 6*\mw,\bsh + 4* \ih - \sh);

\draw[thick, color = color2 , dotted] (2*\gw + 4*\mw,\bsh + 4*\ih + 1*\sh ) -- (2*\gw + 4*\mw,\bsh + 4* \ih + 1*\sh + \ph);

\draw[thick, color = color2] (1*\gw + 1*\mw,\bsh + 5*\ih + 1*\sh + \ph) -- (3*\gw + 7*\mw,\bsh + 5*\ih + 1*\sh + \ph);
\draw[thick, color = color2] (1*\gw + 0*\mw,\bsh + 6*\ih + 1*\sh + \ph) -- (3*\gw + 8*\mw,\bsh + 6*\ih + 1*\sh + \ph);

\draw[thick, color = color1] (1*\gw + 0*\mw,\bsh + 2*\ih) -- (2*\gw + 4*\mw - 0.07,\bsh + 2* \ih);
\draw[thick, color = color1 , dashed] (1*\gw - 1.5*\mw,\bsh + 2*\ih) -- (1*\gw + 0*\mw,\bsh + 2* \ih);
\draw[thick, color = color1] (1*\gw + 0*\mw,\bsh + 3*\ih) -- (2*\gw + 3*\mw,\bsh + 3* \ih);
\draw[thick, color = color1, dashed] (1*\gw - 1.5*\mw,\bsh + 3*\ih) -- (1*\gw + 0*\mw,\bsh + 3* \ih);

\draw[thick, color = color1 , dotted] (1*\gw + 0.5*\mw,\bsh + 3*\ih + 2*\sh) -- (1*\gw + 0.5*\mw,\bsh + 3* \ih + 2*\sh + \ph);

\draw[thick, color = color1] (1*\gw + 0*\mw,\bsh + 4*\ih + 2*\sh + \ph) -- (1*\gw + 2*\mw,\bsh + 4*\ih + 2*\sh + \ph);
\draw[thick, color = color1, dashed] (1*\gw - 1.5*\mw,\bsh + 4*\ih + 2*\sh + \ph) -- (1*\gw + 0*\mw,\bsh + 4*\ih + 2*\sh + \ph);
\draw[thick, color = color1] (1*\gw + 0*\mw,\bsh + 5*\ih + 2*\sh + \ph) -- (1*\gw + 1*\mw,\bsh + 5*\ih + 2*\sh + \ph);
\draw[thick, color = color1, dashed] (1*\gw -1.5*\mw,\bsh + 5*\ih + 2*\sh + \ph) -- (1*\gw + 0*\mw,\bsh + 5*\ih + 2*\sh + \ph);

\end{tikzpicture}
\caption{The bridge and highway gadget. The intervals of eligible machines of highway, bridge and variable jobs are depicted in blue, red and orange, respectively. In this example, variable $x_n$ occurs for the second time in its positive form in the last clause at the first position, and for the first time in its negative form in the first clause at the first position.}
\label{fig:interval}
\end{figure}
\begin{claim}\label{claim:refined_BJob_distribution_detail}
The jobs $\bridgejob_{j,t,j'}^{\true}$ and $\bridgejob_{j,t,j'}^{\false}$ can exclusively be scheduled on $\bridgein_{j,t,j'}$ and $\bridgeout_{j,t,j'}$, for each $j\in [n]$, $j'\in\set{j+1,\dots,n}$ and $t\in [4]$, and each of the two machines receives exactly one of the two jobs.
\end{claim}
\begin{claimproof}
The claim can be proved with a simple inductive argument:
Let $j'\in \set{2,\dots, n}$ and, furthermore, $(j_\ell,t_\ell)\in[j'-1]\times[4]$ denote the $\ell$-th element from $[j'-1]\times[4]$ when ordering the pairs $(j,t)\in[j'-1]\times[4]$ by the increasing lexicographical ordering of the pairs $\kappa(j,t)$.
Considering the ordering of the machines and the job restrictions, $\bridgejob_{j_1,t_1,j'}^{\true}$ and $\bridgejob_{j_1,t_1,j'}^{\false}$ are the only bridge jobs that can be scheduled on $\bridgein_{j_1,t_1,j'}$ and $\bridgeout_{j_1,t_1,j'}$ (see Figure \ref{fig:interval}). 
Hence, the claim has to hold for $(j_1,t_1)$.
But then again $\bridgejob_{j_2,t_2,j'}^{\true}$ and $\bridgejob_{j_2,t_2,j'}^{\false}$ are the only remaining bridge jobs that can be scheduled on $\bridgein_{j_2,t_2,j'}$ and $\bridgeout_{j_2,t_2,j'}$, and so on.
\end{claimproof}
\begin{claim}\label{claim:refined_HJob_distribution_detail}
The jobs $\highwayjob_{j,t,j'}^{\true}$ and $\highwayjob_{j,t,j'}^{\false}$ can exclusively be scheduled on machine $\mathtt{X}$ and $\mathtt{Y}$, for each $j\in [n]$,  $j'\geq j$, and $t\in [4]$; where $\mathtt{X} = \bridgeout_{j,t,j'}$ if $j'>j$, and $\mathtt{X} = \gatemach_{j,t}$ otherwise, and $\mathtt{Y} = \bridgein_{j,t,j'+1}$ if $j'< n$, and $\mathtt{Y} = \clausemach_{\kappa(j,t)}$ otherwise.
Furthermore, each of the two machines receives exactly one of the two jobs.
\end{claim}
\begin{claimproof}
We can use the same argument (with reversed orderings) as we did in the last claim.
It is only slightly more complicated, because more machine types are involved.
\end{claimproof}
Summing up, each job except for clause jobs may only be scheduled on the first or last machine of their interval of eligible machines, and each of these machines receives either the respective job in $\true$- or $\false$-configuration.
Considering this distribution of the jobs and the last digit of the size vectors, we get the following two claims:
\begin{claim}\label{claim:refined_same_truth_configuration}
For any machine, the jobs assigned to this machine all have the same truth configuration (excluding private loads).
\end{claim}
\begin{claim}\label{claim:refined_CMach_triplet_truth_configuration}
For each $i\in [2m]$, the three clause machines corresponding to $i$ receive exactly one highway job with $\true$-configuration, if $C_i$ is a 1-in-3-clause, and exactly two such jobs, if $C_i$ is a 2-in-3-clause.
\end{claim}
The former property together with the possible job distribution determined so far implies that there are only few possible schedules for each machine. 
We summarize these schedules in Table \ref{table:refined_possible_schedules}.
\begin{table}
\centering
\caption{For each machine there are only few possible jobs that may be assigned to it in a schedule with makespan $T$. Each set corresponds to one of the possible schedules.}
\begin{tabular}{ll} 
\toprule
Machine & Possible Schedule \\\midrule
$\truthmach_{j,1}$ & $\set{\truthjob_{j}^{\true}, \variablejob_{j,1}^\true, \variablejob_{j,2}^\true}$, $\set{\truthjob_{j}^{\false}, \variablejob_{j,1}^\false, \variablejob_{j,2}^\false}$\\
$\truthmach_{j,2}$ & $\set{\truthjob_{j}^{\true}, \variablejob_{j,3}^\true, \variablejob_{j,4}^\true}$, $\set{\truthjob_{j}^{\false}, \variablejob_{j,3}^\false, \variablejob_{j,4}^\false}$\\
$\gatemach_{j,t}$ & $\set{\variablejob_{j,t}^{\true}, \highwayjob_{j,t,j}^{\true}}$, $\set{\variablejob_{j,t}^{\false}, \highwayjob_{j,t,j}^{\false}}$\\
$\bridgein_{j,t,j'}$ & $\set{\bridgejob_{j,t,j'}^{\true}, \highwayjob_{j,t,j'-1}^{\true}}$, $\set{\bridgejob_{j,t,j'}^{\false}, \highwayjob_{j,t,j'-1}^{\false}}$\\
$\bridgeout_{j,t,j'}$ & $\set{\bridgejob_{j,t,j'}^{\true}, \highwayjob_{j,t,j'}^{\true}}$, $\set{\bridgejob_{j,t,j'}^{\false}, \highwayjob_{j,t,j'}^{\false}}$ \\
$\clausemach_{i,s}$ (1-in-3) & $\set{\clausejob_{i,1}^{\true}, \highwayjob_{\kappa^{-1}(i,s),n}^{\true}}$, $\set{\clausejob_{i,2}^{\false}, \highwayjob_{\kappa^{-1}(i,s),n}^{\false}}$,$\set{\clausejob_{i,3}^{\false}, \highwayjob_{\kappa^{-1}(i,s),n}^{\false}}$\\
$\clausemach_{i,s}$ (2-in-3) & $\set{\clausejob_{i,1}^{\true}, \highwayjob_{\kappa^{-1}(i,s),n}^{\true}}$, $\set{\clausejob_{i,2}^{\true}, \highwayjob_{\kappa^{-1}(i,s),n}^{\true}}$,$\set{\clausejob_{i,3}^{\false}, \highwayjob_{\kappa^{-1}(i,s),n}^{\false}}$\\\bottomrule
\end{tabular}
\label{table:refined_possible_schedules}
\end{table}
Furthermore, we can infer that the truth assignment gadget works essentially the same as before (see Figure \ref{fig:truth_assignment_gadget}):
\begin{claim}\label{claim:refined_truth_gadget}
Let $j\in[n]$. The truth configuration of any job scheduled on $\truthmach_{j,1}$ is distinct from the truth configuration of any job scheduled on $\truthmach_{j,2}$.
\end{claim}
Lastly, we can show that the bridge and highway gadget works as well:
\begin{claim}\label{claim:refined_signal_is_passed_on}
Let $j\in [n]$ and $t\in[4]$. The variable job scheduled on $\truthmach_{j,\ceil{t/2}}$ and the highway job scheduled on $\clausemach_{\kappa(j,t)}$ have the same truth configuration. 
\end{claim}
\begin{claimproof}
Note that the truth configuration of the variable job scheduled on $\gatemach_{j,t}$ compared with the one of the variable job scheduled on $\truthmach_{j,\ceil{t/2}}$ is reversed. 
Hence, the highway job scheduled on $\gatemach_{j,t}$ also has the reversed truth-configuration while the highway job that is passed on again has the original truth-configuration.
This argument can be repeated with the bridge and highway jobs in the following, yielding the asserted claim.
\end{claimproof}
Using the above claims, we can conclude the proof of Theorem \ref{thm:result_rai} via the following Lemma:
\begin{lemma}
There is a fulfilling truth assignment for the given \tailoredsat instance, if and only if there is a schedule with makespan $T$ for the constructed \rai instance. 
\end{lemma}
\begin{proof}
First, we consider the case that a schedule with makespan $T$ for the constructed \rai instance is given.
For each variable $x_j$ and occurrence $t\in[4]$, let $\highwayjob_{j,t,n}^{\circ_{j,t}}$ be the highway job scheduled on $\clausemach_{\kappa(j,t)}$ (see Table \ref{table:refined_possible_schedules}).
We choose the truth value of $x_j$ to be $\circ_{j,1}$.
Considering the distribution of jobs on the truth assignment machines (see Table \ref{table:refined_possible_schedules}), as well as Claim \ref{claim:refined_truth_gadget} and \ref{claim:refined_signal_is_passed_on}, we know that for each variable $x_j$ and occurrence $t\in[4]$, the truth configuration $\circ_{j,t}$ corresponds exactly to the truth value $x_j$ contributes to the clause given by $\kappa(j,t)$.
Furthermore,  we know that for each clause $C_i$, there are exactly three variable jobs scheduled on the corresponding clause machines, and exactly one or two of these has $\true$-configuration, if $C_i$ is a 1-in-3-clause or 2-in-3-clause, respectively (Claim \ref{claim:refined_CMach_triplet_truth_configuration}).
Hence, $C_i$ is fulfilled.

Now, let there be a fulfilling truth assignment, and $\triangleleft_j$ be the corresponding truth value of variable $x_j$ and $\triangleright_j$ its negation.
We set $\bigtriangleup_{j,t} = \triangleleft_j$ for $t\in\set{1,2}$ and $\bigtriangleup_{j,t} = \triangleright_j$ for $t\in\set{3,4}$ and assign $\highwayjob_{j,t,n}^{\bigtriangleup_{j,t}}$ to $\clausemach_{\kappa(j,t)}$.
Let $\bigtriangledown_{jt}$ be the negation of $\bigtriangleup_{jt}$.
All the other jobs are assigned as indicated by the claims and Table \ref{table:refined_possible_schedules} in particular:
Each machine receives its private load;
$\clausemach_{\kappa(j,t)}$ additionally receives one of the eligible remaining clause jobs with $\bigtriangleup_{j,t}$-configuration (this can be done because the truth assignment is fulfilling); 
$\bridgeout_{j,t,j'}$ receives $\highwayjob_{j,t,j'}^{\bigtriangledown_{j,t}}$ and $\bridgejob_{j,t,j'}^{\bigtriangledown_{j,t}}$;
$\bridgein_{j,t,j'}$ receives $\highwayjob_{j,t,j'-1}^{\bigtriangleup_{j,t}}$ and $\bridgejob_{j,t,j'}^{\bigtriangleup_{j,t}}$;
$\gatemach_{j,t}$ receives $\highwayjob_{j,t,j}^{\bigtriangledown_{j,t}}$ and $\variablejob_{j,t}^{\bigtriangledown_{j,t}}$;
$\truthmach_{j,1}$ receives $\variablejob_{j,1}^{\bigtriangleup_{j,1}}$, $\variablejob_{j,2}^{\bigtriangleup_{j,2}}$ and $\truthjob_j^{\bigtriangleup_{j,1}}$ ; and $\truthmach_{j,2}$ receives $\variablejob_{j,3}^{\bigtriangleup_{j,3}}$, $\variablejob_{j,4}^{\bigtriangleup_{j,4}}$ as well as $\truthjob_j^{\bigtriangleup_{j,3}}$.
It is easy to verify, that all jobs are assigned and each machine has a load of $T$.
\end{proof}

\section{Resource Restrictions}
\label{sec:resource}

In this section, we first present some preliminary observations concerning \rar{R} and discuss the relationship of the problem with \rai and \lrs{D}.
Next, we revisit established reductions for the restricted assignment problem and show that they can be modeled with only few resources.
This already gives the result for $4$ resources in \cref{thm:result_rar}.
Lastly, we study the cases with $2$ and $3$ resources. 
We first give a reduction for $R=3$ and then refine the result to work for $R=2$ as well thereby concluding the proof of \cref{thm:result_rar} and \cref{thm:results_santa}.

\subparagraph{Preliminaries.}

Recall that in the problem of scheduling with resource restrictions with $R$ resources (\rar{R}), a set $\ress$ of $R$ (renewable) resources is given, each machine $i$ has a resource capacity $c_{r}(i)$ and each job $j$ has a resource demand $d_{r}(j)$ for each $r\in\ress$.
Job $j$ is eligible on machine $i$, if $d_{r}(j)\leq c_{r}(i)$ for each resource~$r$. 
We allow arbitrary real values for the capacities and demands but it is not hard to see that relatively small integer values suffice.
Indeed, given an instance of \rar{R}, we may perform the following two steps:
First, we increase for each job $j$ and each resource $r$ the demand $d_r(j)$ to the smallest value included in $\sett{c_r(i)}{i\in\machs,c_r(i)\geq d_r(j)}$ (if this set is empty the job cannot be processed anywhere).
Afterwards, there are at most $m=|\machs|$ distinct demand or capacity values for each resource, and we can change the smallest value to $1$ the second to $2$ and so on.
This yields an instance with the same restrictions and the property that all the capacities and demands are included in $[m]$.

Technically, there are two versions of the problem \rar{R} depending on whether the resources, demands and capacities are explicitly given or not. 
In the second variant recognition is an issue that we do not address in this work since our results work for both versions of the problem.
However, note that the proof of the following lemma gives some intuition concerning this: 
\begin{lemma}\label{lem:RARm}
Each restricted assignment instance with $m$ machines is also a \rar{m} instance; and for each $m\in \NN$ there is a \rar{m} instance with $m$ machines (and $m$ jobs) that is not a \rar{R} instance for any $R<m$.
\end{lemma}
\begin{proof}
Given a restricted assignment instance with $m$ machines, we define resources, demands and capacities that model the given restrictions.
First, we identify each machine with a resource, that is, we set $\ress = \machs$.
Furthermore, we set the capacities of machine $i\in\machs$ concerning resource $r\in\ress$ to be $c_r(i) = 1$ if $r\neq i$ and $c_i(i) = 0$; and the demand of job $j\in\jobs$ concerning $r$ to be $d_r(i) = 0$ if $r\in\emachs(j)$, and $d_r(i) = 1$ otherwise.
It is easy to check that for each job $j$ and machine $i$, we have $d_r(j)\leq c_r(i)$ for each resource $r$, if and only if $i\in\emachs(j)$. 

The next goal is to construct a simple \rar{m} instance for each $m\in\NN$.
We first collect some simple observations.
Given some instance of \rar{R} with $j\in\jobs$ and $i'\in\machs\setminus\emachs(j)$, we know that there is a resource $r(j,i')\in\ress$ such that $c_{r(j,i')}(i') < d_{r(j,i')}(j)$ (otherwise $i'\in\emachs(j)$).
One could say that $r(j,i')$ separates $i'$ from $j$ or $\emachs(j)$.
On the other hand, we know that $d_{r}(j) \leq c_{r}(i)$ for each $i\in\emachs(j)$ and $r\in\ress$.
Let $i\in\emachs(j)$.
Now, consider the case that we have another job $j'$ with $ i'\in\emachs(j') $ and $i\in\machs\setminus\emachs(j')$.
Then we have $r(j,i') \neq r(j',i)$, because otherwise:
\[d_{r(j',i)}(j) \leq c_{r(j',i)}(i) < d_{r(j',i)}(j') \leq c_{r(j',i)}(i') = c_{r(j,i')}(i') < d_{r(j,i')}(j) = d_{r(j',i)}(j)\] 
Using this insight, we construct the instance as follows:
We set $\machs = [m]$ and $\jobs = \sett[\big]{j\subseteq\machs}{|j| = m-1}$ with $\emachs(j) = j$. 
We may assume unit processing times.
This instance has exactly $m$ jobs and machines and we know that it is a \rar{m} instance because of the first part of the proof.
Now, given any resources $\ress$ along with capacities and demands that model the restrictions for the above instance, we show that $|\ress|\geq m$.
For each $j\in\jobs$ let $i_j\in\machs$ be the single machine that is restricted to process $j$, i.e., $i_j\in\machs\setminus\emachs(j)$.
This implies $\machs = \sett{i_j}{j\in\jobs}$ and due to the above observation, we have $r(j,i_j) \neq r(j',i_{j'})$ for each pair of distinct jobs $j,j'\in\jobs$.
Hence, $\sett{r(j,i_j)}{j\in\jobs}= \ress'\subseteq \ress$ and $|\ress'| = m$ concluding the proof.
\end{proof}

The relationship between scheduling with resource and interval restrictions is discussed in the following lemma:
\begin{lemma}
Each \rar{1} instance is also a \rai instance and there is a \rai instance that is not a \rar{1} instance.
Moreover, each \rai instance is also a \rar{2} instance and there is a \rar{2} instance that is not a \rai instance.
With a slight abuse of notation, we may write: $\text{\rar{1}}\subset\text{\rai}\subset\text{\rar{2}}$.
\end{lemma}
\begin{proof}
Given an instance of \rar{1}, we may sort the machines based on their capacity values decreasingly. 
Than each job $j$ that can be processed on any machine, can also be processed on any predecessor of this machine.
Hence, $\emachs(j)$ corresponds to an interval of machines starting with the first machine.
On the other hand, consider an instance with two machines and two jobs.
The first job is (exclusively) eligible on the first machine and the second one on the second.
This instance is a \rai but not a \rar{2} instance.

Given an instance of \rai, we may assume wlog.\xspace $\machs=[m]$ and that the ordering of the machines is the natural ordering. 
We set $\ress = [2]$.
Furthermore, for each machine $i$, we set $c(i)= (i,(m+1)-i)$; and for each job $j$ with $\emachs(j)=\set{\ell,\dots,r}$, we set $d(j) = (\ell, (m+1) - r)$ (see Figure \ref{fig:2res}).
If $ i \in \emachs(j)$, we have $c_1(i) = i \geq \ell = d_1(j)$ and $c_2(i) = (m+1) - i \geq  (m+1) - r = d_2(j)$; and if $ i \in \machs\setminus\emachs(j)$, we have either $i< \ell$ or $i > r$ which implies $c_1(i) < d_1(j)$ or $c_2(i) < d_2(j)$ respectively.
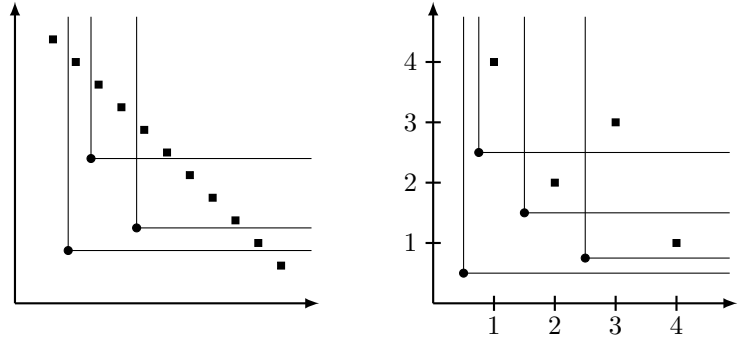
\begin{figure}
\centering
\begin{tikzpicture}
\pgfmathsetmacro{\ch}{4}
\pgfmathsetmacro{\cw}{4}

\draw[thick, ->, >=latex] (0,0) -- (0,\ch);
\draw[thick, ->, >=latex] (0,0) -- (\cw,0);

\foreach \x in {0,...,10}{
    \filldraw[xshift=-1.3pt,yshift=-1.3pt] ({(\x * 0.5 + (10-\x )*(\cw - 0.5))/10},{(\x * (\ch - 0.5) + (10-\x )*0.5)/10}) rectangle ++(2.6pt,2.6pt);
    }

\filldraw (1, 0.48*\ch) circle (1.5pt);
\draw (1, 0.48*\ch) -- (1, \ch - 0.2);
\draw (1, 0.48*\ch) -- (\cw - 0.1, 0.48*\ch);

\filldraw (0.7, 0.7) circle (1.5pt);
\draw (0.7, 0.7) -- (0.7, \ch - 0.2);
\draw (0.7, 0.7) -- (\cw - 0.1, 0.7);

\filldraw (0.4 * \cw, 1.0) circle (1.5pt);
\draw (0.4 * \cw, 1.0) -- (0.4 * \cw, \ch - 0.2);
\draw (0.4 * \cw, 1.0) -- (\cw - 0.1, 1.0);

\begin{scope}[xshift = 5.5cm]

\draw[thick, ->, >=latex] (0,0) -- (0,\ch);
\draw[thick, ->, >=latex] (0,0) -- (\cw,0);

\foreach \x in {1,...,4}{
    \draw[thick] ({(\x * (0.8*\cw)/4}, -0.1) -- ({(\x * (0.8*\cw)/4}, 0.1) node [midway, yshift = -0.3cm, text = black] {\x};
    \draw[thick] (0.1,{(\x * (0.8*\ch)/4}) -- (-0.1, {(\x * (0.8*\ch)/4}) node [midway, xshift = -0.3cm, text = black] {\x};
    }
\filldraw ({(0.5 * (0.8*\cw))/4}, {(0.5 * (0.8*\ch))/4}) circle (1.5pt);
\draw ({(0.5 * (0.8*\cw))/4}, {(0.5 * (0.8*\ch))/4}) -- ({(0.5 * (0.8*\cw))/4}, \ch - 0.2);
\draw ({(0.5 * (0.8*\cw))/4}, {(0.5 * (0.8*\ch))/4}) -- (\cw - 0.1, {(0.5 * (0.8*\ch))/4});

\filldraw ({(2.5 * (0.8*\cw))/4}, {(0.75 * (0.8*\ch))/4}) circle (1.5pt);
\draw ({(2.5 * (0.8*\cw))/4}, {(0.75 * (0.8*\ch))/4}) -- ({(2.5 * (0.8*\cw))/4}, \ch - 0.2);
\draw ({(2.5 * (0.8*\cw))/4}, {(0.75 * (0.8*\ch))/4}) -- (\cw - 0.1,  {(0.75 * (0.8*\ch))/4});

\filldraw ({(1.5 * (0.8*\cw))/4}, {(1.5 * (0.8*\ch))/4}) circle (1.5pt);
\draw ({(1.5 * (0.8*\cw))/4}, {(1.5 * (0.8*\ch))/4}) -- ({(1.5 * (0.8*\cw))/4}, \ch - 0.2);
\draw ({(1.5 * (0.8*\cw))/4}, {(1.5 * (0.8*\ch))/4}) -- (\cw - 0.1,  {(1.5 * (0.8*\ch))/4});

\filldraw ({(0.75 * (0.8*\cw))/4}, {(2.5 * (0.8*\ch))/4}) circle (1.5pt);
\draw ({(0.75 * (0.8*\cw))/4}, {(2.5 * (0.8*\ch))/4}) -- ({(0.75 * (0.8*\cw))/4}, \ch - 0.2);
\draw ({(0.75 * (0.8*\cw))/4}, {(2.5 * (0.8*\ch))/4}) -- (\cw - 0.1,  {(2.5 * (0.8*\ch))/4});


\filldraw[xshift=-1.3pt,yshift=-1.3pt] ({(3 * (0.8*\cw))/4}, {(3 * (0.8*\ch))/4}) rectangle ++(2.6pt,2.6pt);
\filldraw[xshift=-1.3pt,yshift=-1.3pt] ({(4 * (0.8*\cw))/4}, {(1 * (0.8*\ch))/4}) rectangle ++(2.6pt,2.6pt);
\filldraw[xshift=-1.3pt,yshift=-1.3pt] ({(2 * (0.8*\cw))/4}, {(2 * (0.8*\ch))/4}) rectangle ++(2.6pt,2.6pt);
\filldraw[xshift=-1.3pt,yshift=-1.3pt] ({(1 * (0.8*\cw))/4}, {(4 * (0.8*\ch))/4}) rectangle ++(2.6pt,2.6pt);

\end{scope}

\end{tikzpicture}
\caption{The left picture visualizes that each \rai instance can be seen as a \rar{2} instance and the right one depicts an \rar{2} instance that is not a \rai instance. In both pictures, each dimension corresponds to a resource, the squares mark the capacities of machines and the circles the demands of jobs. If the capacity of a machine is at least as big as the demand of a job in both dimension, the job is eligible on the machine.}
\label{fig:2res}
\end{figure}

Lastly, we construct an instance of \rar{2} that is not a \rai instance.
Let $\machs = [4]$, $\ress = [2]$ and $\jobs = \set[\big]{\set{1,2,3,4}, \set{1,2}, \set{1,3}, \set{1,4}}$.
We may assume unit job sizes.
The resource capacities and demands are given in Table \ref{table:3res_rar2_neq_rai} and the construction is illustrated in Figure \ref{fig:2res}.
It is easy to see that $\emachs(j) = j$ for each $j\in\jobs$.
In any total ordering of the machines in which each job is eligible on consecutive machine, the machine $1$ has to be a direct neighbor of $2$, $3$ and $4$. 
This is not possible.
\begin{table}
\centering
\caption{A simple example of a \rar{2} instance that is not a \rai instance. Note that all demands could be rounded up to the next integer value without changing the construction.}
\begin{tabular}{ll|ll}
\toprule
Machine & Capacity & Job & Demand \\
\midrule
$1$ & $(3,3)$ & $\set{1,2,3,4}$ & $(0.5,0.5)$\\
$2$ & $(4,1)$ & $\set{1,2}$ & $(2.5,0.75)$\\
$3$ & $(2,2)$ & $\set{1,3}$ & $(1.5,1.5)$\\
$4$ & $(1,4)$ & $\set{1,4}$ & $(0.75,2.5)$\\
\bottomrule
\end{tabular}
\label{table:3res_rar2_neq_rai}
\end{table}
\end{proof}

We already mentioned in the introduction that there is a close relationship between scheduling with resource restrictions and low rank unrelated scheduling.
Remember that in the rank $D$ unrelated scheduling problem (\lrs{D}) the processing time matrix $(p_{ij})_{i\in\machs,j\in\jobs}$ has a rank of at most $D$, or, equivalently, there is a $D$ dimensional size or speed vector $s(j)$ or $v(i)$ for each job $j$ or machine $i$ respectively, and the processing time $p_{ij}$ is given by $\sum_{d\in[D]}s_d(j)v_d(i)$.
It is easy to construct for any $m\in\NN$ a \rar{1} instance that is a \lrs{m} instance as well:
The instance with $\jobs=\machs=[m]$, $\ress = \set{1}$, $d_1(k) = c_1(k) = k$ for each $k\in[m]$ and unit processing times suffices (assuming that the number $\infty$ is interpreted as some sufficiently big number).
On the other hand, any \rar{R} instance can be approximated with arbitrary precision by \lrs{R+1} instances in the following sense:
\begin{lemma}
Let $I$ be a \rar{R} instance. 
For any $\eps,K > 0$, there is a \lrs{R+1} instance $I'$ with the same jobs and machines and the following property: 
Let $p'_{ij}$ be the processing time of job $j$ on machine $i$ in instance $I'$. 
We have $p_j \leq p'_{ij} \leq p_j + \eps$ if $i\in\emachs(j)$ and $p'_{ij}\geq p_j + K$ otherwise.
\end{lemma}
\begin{proof}
Wlog. we assume $\ress = [R]$.
Let $\delta = \eps / R$ and $N = \max\set{ K/\delta , 1}$.
We define the size and speed vectors of $I'$ as follows:
For each job $j$ we set $s'_r(j) = \delta N^{d_r(j)}$ for each $r\in[R]$, as well as $s'_{R+1}(j) = p_j$.
Moreover, for each machine $i$ we set $v'_r(j) = N^{-c_r(j)}$ for each $r\in[R]$, as well as $v'_{R+1}(j) = 1$.
Then $p'_{ij} = \sum_{r\in[R+1]}s_r(j)v_r(i) = p_j + \sum_{r\in[R]}\delta N^{d_r(j)-c_r(j)}$ and therefore $p'_{ij}\geq p_j$ in any case.
Furthermore, if $i\in\emachs(j)$, we have $d_r(j) \leq c_r(j)$ for each $r\in[R]$, and hence $p'_{ij} \leq p_j + R\delta = p_j + \eps$.
If, on the other hand, $i\not\in\emachs(j)$, then there is an $r\in[R]$ such that $d_r(j) > c_r(j)$ yielding $p'_{ij} \geq p_j + \delta N \geq p_j + K$.
\end{proof}
The above lemma implies that from the perspective of approximation algorithms \rar{R} is essentially included in \lrs{R+1}.
We could use this lemma and \cref{thm:result_rai} or \cref{thm:result_rar} to show that there is no PTAS for \lrs{3} unless P$=$NP.
While this is already known \cite{CMYZ17}, the resulting construction may be more accessible.

\subparagraph{Established Reductions Revisited.}

In the following, we first present the classical reduction by Lenstra et al. \cite{LenstraST90} showing $1.5$-inapproximability for the restricted assignment problem.
We show that the restricted assignment instances in this reduction can be modeled using $6$ resources yielding the same hardness for \rar{6}.
Nearly the same argument was used by Bhaskara et al.~\cite{BhaskaraKTW13} to show $1.5$-inapproximability for \lrs{7}.
Next, we take the same approach for the more recent reduction by Ebenlendr et al. \cite{EbenlendrKS14} and show $1.5$-inapproximability already for $4$ resources.

In the \threedm problem, the input consists of three disjoint sets $A$, $B$ and $C$ with $|A|=|B|=|C| = n\in\NN$, as well as a set of triplets $E \subseteq \sett[\big]{\set{a,b,c}}{a\in A,b\in B, c\in C}$.
The goal is to decide whether there is a subset $F\subseteq E$ that perfectly covers $A$, $B$ and $C$, that is, for each $x\in A\cup B\cup C$ there is exactly one triplet $e\in F$ with $x\in e$.
The set $F$ is called a 3D-matching.
We assume that the elements of $A$, $B$ and $C$ are indexed, that is, $A = \set{a_1,a_2,\dots, a_n}$, $B= \set{b_1,b_2,\dots, b_n}$ and $C= \set{c_1,c_2,\dots, c_n}$.
Furthermore, we assume that for each $x\in A\cup B\cup C$ there is at least one $e\in E$ with $x\in E$ (otherwise the problem is trivial).
Via a reduction from \threedm to the restricted assignment problem, Lenstra et al. \cite{LenstraST90} showed:
\begin{theorem}[\cite{LenstraST90}]\label{prop:LSTreduction}
There is no polynomial time approximation algorithm for restricted assignment with rate smaller than $1.5$ unless P$=$NP.
\end{theorem}
\begin{proof}
Given an instance of \threedm, we set $\machs = E$ and $E(x) = \sett{e\in E}{x \in e}$ for each $x\in A\cup B\cup C$. 
For each $a\in A$, we introduce $|E(a)| - 1 \geq 0$ many dummy jobs with size $2$ and eligible on machines $e\in E(a)$.
Moreover, for each $x\in B\cup C$ we introduce an element job with size $1$ eligible on machines $e\in E(x)$.
Note that the overall size of the jobs is given by $2n + 2\sum_{a\in A}(|E(a)| - 1) = 2n + 2(|E| - n) = 2|\machs|$.

If there is a schedule with makespan $2$ for this instance, then each machine either processes two element or one dummy job.
For each $x\in B\cup C$, we have $x\in e$ for the machine $e$ processing the corresponding element job, and, furthermore, for each $a\in A$, there is exactly one machine $e$ with $a\in e$ that does not process a dummy job (and therefore processes element jobs).
Hence, we get a 3D-matching by selecting the machines that process element jobs.

If, on the other hand, there is a 3D-matching $F$ for the \threedm instance, than we can schedule the element job corresponding to $x\in B\cup C$ on the machine $e\in F$ with $x\in E$ and the dummy jobs corresponding to $a\in A$ on the $|E(a)| - 1$ machines $e\in E(a)\setminus F$.
This yields a schedule with makespan 2.
\end{proof}
We reproduce the restrictions in the above reduction using six resources and get:
\begin{corollary}
There is no polynomial time approximation algorithm for \rar{6} with rate smaller than $1.5$ unless P$=$NP.
\end{corollary}
\begin{proof}
We set $\ress = \sett[\big]{(X,k)}{X\in\set{A,B,C},k\in[2]}$.
Let $e\in E$ and $X\in\set{A,B,C}$.
We set the resource capacities $c_{(X,1)}(e) = i$ and $c_{(X,2)}(e) = (n+1) - i$.
Let $x_j$ be the element with index $j$ in $X$.
We set the resource demand of a (element or dummy) job $J$ corresponding to $x_j$ as follows:
$d_{(X,1)}(J) = j$, $d_{(X,2)}(J) = (n+1) - j$, as well as $d_{(Y,k)}(J) = 0$ for each $Y\in\set{A,B,C}\setminus\set{X}$ and $k\in[2]$.
It is easy to see that $J$ can exclusively be scheduled on machines $e$ with $x_j\in e$.
\end{proof}

In the classical \threesat problem, a conjunction of $m$ clauses is given and each clause is a disjunction of at most three literals of variables $x_1,\dots,x_n$. 
In the result due to Ebenlendr et al. \cite{EbenlendrKS14}, the modified \threesat problem, where each variable occurs exactly three and each literal at most two times in the formula, is reduced to the graph balancing problem, that is, restricted assignment with the additional property that each job is eligible on at most two machines. 
To show that the modified \threesat problem is NP-hard, we can use techniques already applied in Section \ref{sec:interval}:
We may replace the $d_j$ occurrences of variable $x_j$ with new variables $z_{j1}, \dots, z_{jd_j}$ and add new clauses $(z_{j1} \wedge \neg z_{j2})$, \dots $(z_{jd_{j-1}} \wedge \neg z_{jd_j})$, $(z_{jd_j} \wedge \neg z_{j1})$.
\begin{theorem}[\cite{EbenlendrKS14}]\label{thm:EKSreduction}
There is no polynomial time approximation algorithm with rate smaller than $1.5$ for the graph balancing problem unless P$=$NP.
\end{theorem}
\begin{proof}
Given an instance of modified \threesat, we introduce clause machines $v_i$ corresponding to the clauses $C_i$, and literal machines $u_{j,1}$ and $u_{j,0}$ corresponding to the literals $x_j$ and $\neg x_j$.
Furthermore, we introduce truth assignment jobs $e_j$ for each variable $x_j$ with size $2$ and eligible on $u_{j,1}$ and $u_{j,0}$; and clause jobs $f_{i,j,\alpha}$ for each clause $C_i$ and literal $y_j$ occurring in $C_i$ with $\alpha=1$ if $y_j = x_j$ and $\alpha = 0$ if $y_j = \neg x_j$.
The job $f_{i,j,\alpha}$ has size $1$ and is eligible on $v_i$ and $u_{j,\alpha}$.
Lastly, we introduce a dummy job $d_i$ for each clause $C_i$ with less then three literals.
Its size is $1$ if $C_i$ contains two literals, and $2$ if $C_i$ contains only one literal.

In a schedule with makespan $2$, there is at least one clause job $f_{i,j,\alpha}$ for each $v_i$ that is scheduled on $u_{j,\alpha}$ and not on $v_i$.
Hence, the job $e_j$ has to be scheduled on $u_{j,|\alpha-1|}$.
Now, it is easy to see that there is a schedule with makespan $2$, if and only if there is a fulfilling assignment.
The construction works as follows:
Given a schedule with makespan $2$, we set variable $x_j$ to $\true$ if $e_j$ is scheduled on $u_{j,0}$, and to $\false$ otherwise.
Moreover, given a fulfilling truth assignment we assign the truth assignment jobs correspondingly, and the machines $u_{j,\alpha}$ that did not receive a truth assignment job receive all eligible clause jobs (at most two).
\end{proof}
We reproduce the restrictions in the above reduction using four resources and get:
\begin{corollary}
There is no polynomial time approximation algorithm for \rar{4} with rate smaller than $1.5$ unless P$=$NP.
\end{corollary}
\begin{proof}
We set $\ress = [4]$.
The clause machine $v_i$ has a resource capacity vector of $(2n + 1,2n + 1,i,(m+1) - i)$, and the literal machine $u_{j,\alpha}$ has capacity vectors $(2j - \alpha,(2n + 1) - (2j-\alpha),m+1,m+1)$.
Furthermore, the truth assignment job $e_j$ has a resource demand vector of $(2j - 1, (2n + 1) - 2j , m+1, m+1)$; the clause job $f_{i,j,\alpha}$ has a demand vector of $(2j - \alpha, (2n + 1) - (2j-\alpha) , i,(m+1) - i)$; and the dummy job $d_i$ has a demand vector of $(2n + 1, 2n + 1, i,(m+1) - i)$.
It is easy to verify that the resulting sets of eligible machines are the same as described in \cref{thm:EKSreduction}.
\end{proof}

\subparagraph{Three Resources.}

We present a reduction from \threedm to \rar{3}. 
The reduction is based on the classical result by Lenstra et al. \cite{LenstraST90} and very similar to a reduction by Bhaskara et al.~\cite{BhaskaraKTW13} for \lrs{4}.
However, there is a problem with the choice of processing times in the latter reduction (see \cref{appendix_reduction}), and the present result can be used to fix it.

Given an instance $(A,B,C,E)$ of \threedm, let $n=|A|$ and $E(x) = \sett{e\in E}{x \in e}$ for each $x\in A\cup B\cup C$.
Furthermore, we set $\alpha_A = 12$, $\alpha_B = 13 $, $\alpha_C = 22$, $\beta_A = 14$, $\beta_B = 15$ and $\beta_C = 18$.
Let $\ress = \set{A,B,C}$ and $\machs = E$.
For each machine $e$, we define the resource capacities as follows.
Let $X\in\set{A,B,C}$ and $x_i\in X\cap e$ be the element of $x$ with index $i$.
We set $c_{X}(e) = i$.
Furthermore, for each element $x_i \in X$ with index $i$ in $X\in\set{A,B,C}$, we introduce one element job with size $\alpha_X$ and $|E(x)| - 1$ dummy jobs with size $\beta_X$.
The resource demand for each of these jobs is given by $d(i)$ with $d_{X}(i) = i$ and $d_Y(i) = 0$ for $Y\in\set{A,B,C}\setminus\set{X}$. 

\begin{claim}\label{claim:3res_properties_of_values}
We have $\alpha_A + \alpha_B + \alpha_C = 47 = \beta_A + \beta_B + \beta_C$; 
any four numbers taken from $\Gamma = \set{\alpha_A,\alpha_B,\alpha_C,\beta_A,\beta_B,\beta_C} = \set{12,13,22,14,15,18}$ sum up to a value bigger than $47$; 
any selection of less than $3$ numbers sums up to a value smaller than $47$;
and for any three numbers $\gamma_1,\gamma_2,\gamma_3\in \Gamma$ with $\gamma_1\leq\gamma_2\leq\gamma_3$ and $\gamma_1 + \gamma_2 + \gamma_3 = 47$, we have either $(\gamma_1,\gamma_2,\gamma_3) = (\alpha_A,\alpha_B,\alpha_C)$ or $(\gamma_1,\gamma_2,\gamma_3) = (\beta_A,\beta_B,\beta_C)$.
\end{claim}
\begin{claimproof}
The first three assertions are obvious, and the fourth holds due to a simple case analysis:
\begin{itemize}
\item If $\gamma_1 > 15$, we have $\gamma_1 \geq 18$, and hence $47 = \gamma_1 + \gamma_2 + \gamma_3 \geq 3\cdot\gamma_1 = 54$: a contradiction.
\item Note that $\gamma_3 \geq (\gamma_2 + \gamma_3)/2 = (47 - \gamma_1)/2$. 
Hence, $ \gamma_1\leq 15$ implies $\gamma_3\geq 16$ and therefore $\gamma_3\in\set{18,22}$.
\item If we have $\gamma_3 = 22 = \alpha_{C}$, then $\gamma_1 \leq (\gamma_1 + \gamma_2)/2 = (47 - \gamma_3)/2 = 12.5$.
Hence, $\gamma_1 = 12 = \alpha_{A}$ and $\gamma_2 = 13 = \alpha_{B}$.
\item If we have $\gamma_3 = 18 = \beta_{C}$, than $\gamma_2 \geq (\gamma_1 + \gamma_2)/2 = (47 - \gamma_3)/2 = 14.5$.
Hence, $\gamma_2 \in\set{15,18}$. 
If $\gamma_2 = 15 = \beta_{B}$, then $\gamma_1 = 14 = \beta_{A}$, and if $\gamma_2 = 18$, then $\gamma_1 = 11\notin\Gamma$.
\end{itemize}
This concludes the proof of the claim.
\end{claimproof}
By brute force, it can be verified that $47$ is the smallest value such that suitable numbers $\alpha_A$, $\alpha_B$, $\alpha_C$, $\beta_A$, $\beta_B$ and $\beta_C$ exist and the above claim holds.
\begin{claim}\label{claim:3res_numbermachs*T=job_load}
The summed up size of all the element and dummy jobs is $47|\machs|$.
\end{claim}
\begin{claimproof}
We have exactly $n$ element jobs with size $\alpha_A$, $\alpha_B$ and $\alpha_C$, respectively, yielding an overall load of $47 n$.
The dummy jobs have an overall load of:
\begin{align*}
&\beta_A\sum_{a\in A}(|E(a)| - 1) + \beta_B\sum_{b\in B}(|E(b)| - 1) +\beta_C\sum_{c\in C}(|E(b)| - 1)\\
=&(\beta_A + \beta_B + \beta_C) (|E| - n) = 47(|\machs| - n)
\end{align*}
In this equation, we used the simple fact that $\sett{E(x)}{x\in X}$ is a partition of $E$ for each $X\in\set{A,B,C}$, and hence $|E| = \sum_{x\in X}|E(x)|$.
\end{claimproof}
These two claims imply:
\begin{claim}\label{claim:3res_3jobspermach}
In any schedule for the constructed instance with makespan $47$, each machine receives exactly three jobs with sizes $\gamma_1,\gamma_2,\gamma_3$ such that $(\gamma_1,\gamma_2,\gamma_3) = (\alpha_A,\alpha_B,\alpha_C)$ or $(\gamma_1,\gamma_2,\gamma_3) = (\beta_A,\beta_B,\beta_C)$.
\end{claim}
Using these claims, we can show:
\begin{proposition}
There is a perfect matching for the given \threedm instance, if and only if there is a schedule with makespan $47$ for the constructed \rar{3} instance. 
\end{proposition}
\begin{proof}
Let $F$ be a perfect matching for the \threedm instance.
For each $x\in A\cup B\cup C$ we assign the corresponding element job to the machine $e$ with $x\in e$ and $e\in F$.
Furthermore, the dummy jobs corresponding to $x\in X$ with $X\in\set{A,B,C}$, are distributed to the machines $e$ with $x\in e$ and $e\notin F$ such that each machine receives exactly one job.
Hence, each machine $e\in E$ receives exactly three eligible jobs either with sizes $\alpha_A$, $\alpha_B$ and $\alpha_C$ (if $e\in F$) or $\beta_A$, $\beta_B$ and $\beta_C$ (otherwise). 

Next, we assume that there is a schedule with makespan $47$ for the scheduling instance.
For each $X\in \set{A,B,C}$, there are exactly $|\machs|$ many jobs with size $\alpha_X$ or $\beta_X$, and due to the above claims, we know that each machine receives exactly one of these jobs.
For each $j\in [n]$, let $x_j\in X$ be the element with index $j$ in $X\in\set{A,B,C}$.
The machines $\bigcup_{j=i}^{n}E(x_j)$ are the only machines that may process jobs corresponding to $x_i,\dots, x_n$ for each $i\in[n]$ and we have exactly $\sum_{j=i}^{n} |E(x_j)|$ many such jobs.
Hence, the machines from $E(x_i)$ receive exactly the jobs corresponding to $x_i$.
Now, considering this and Claim \ref{claim:3res_3jobspermach}, we get a perfect matching by selecting the machines that process three element jobs.
\end{proof}

\subparagraph*{Two Resources.}

We are able to refine the result for three resources to work for two resources as well by using another variant of \threedm as the starting point of the reduction.
The problem \tailoredmatching was introduced by Chen et al. \cite{ChenYZ14} to get an improved lower bound for the approximation ratio of rank four unrelated scheduling.
In this problem, a set of six disjoint sets $\mathcal{E} = \set{A,A',B,B',C,C'}$ is given.
For each $X\in \mathcal{E}$, we have $|X| = 3n$ for some $n\in\NN$ and the sets are indexed by $[3n]$, e.g., $A = \set{a_1,a_2,\dots, a_{3n}}$.
Furthermore, there are two sets of triplets $E_1\subseteq\sett[\big]{\set{a_i,b_j,c_j},\set{a'_i,b_j,c_j}}{i\in[3n], j\in[3n]}$ and $E_2 = \sett[\big]{\set{a_i,b'_i,c'_i},\set{a'_i,b'_i,c'_{\zeta(i)}}}{i\in[3n]}$ with $\zeta(3k + 1) = 3k + 2$, $\zeta(3k + 2) = 3k + 3$ and $\zeta(3k + 3) = 3k + 1$ for each $k\in\set{0,\dots,n-1}$.
Note that the second set of triplets is already determined by the element sets in the input.
Similar to the classical \threedm problem, the goal is to decide whether there is a subset $F\subseteq E_1\cup E_2$ that perfectly covers the element set, that is, for each $x\in \bigcup_{X\in\mathcal{E}}X$ there is exactly one triplet $e\in F$ with $x\in e$.
Furthermore, we assume that for each $x\in \bigcup_{X\in\mathcal{E}}X$ there is at least one $e\in E$ with $x\in E$ (otherwise the problem is trivial).

Let $\alpha_A = \alpha_{A'} = 12$, $\alpha_B = \alpha_{B'} = 13 $, $\alpha_C= \alpha_{C'}=22$, $\beta_A = \beta_{A'} =14$, $\beta_B = \beta_{B'} =15$ and $\beta_C = \beta_{C'} = 18$.
We set $\machs = E_1\cup E_2$ and $\ress = [2]$.
The corresponding resource capacity vectors are presented in Table \ref{table:3res_res_values}.
Furthermore, for each element $x \in X$ in $X\in\mathcal{E}$, we introduce one element job with size $\alpha_X$ and $|E(x)| - 1$ dummy jobs with size $\beta_X$.
The vector of resource demands for each such job is given in Table \ref{table:3res_res_values}.
\begin{table}
\centering
\caption{The resource demands and capacities for the different job (types) and machines.}
\begin{tabular}{ll|ll}
\toprule
Jobs & Resources & Machines & Resources \\\midrule
$a_i$ & $(2i,0)$ & $\set{a_i,b_j, c_j}$ & $(2i, 3n + j)$\\
$a'_i$ & $(2i - 1,0)$ & $\set{a'_i,b_j, c_j}$ & $(2i-1, 3n + j)$\\
$b_j$ & $(0, 3n + j)$ & $\set{a_{i},b'_{i}, c'_{i}}$ & $(2i, i)$\\
$c_j$ & $(0, 3n + j)$ & $\set{a'_{i},b'_{i}, c'_{\zeta(i)}}$ & $(2i - 1, \zeta(i))$\\
$b'_i$ & $(2i - 1, 0)$ & &\\
$c'_{i}$ & $(0, i)$ & &\\
\bottomrule
\end{tabular}
\label{table:3res_res_values}
\end{table}
Note that Claim \ref{claim:3res_properties_of_values}, \ref{claim:3res_numbermachs*T=job_load} and \ref{claim:3res_3jobspermach} hold for this reduction as well and with the same reasoning.
Furthermore, a simple case analysis yields:
\begin{claim}
For each $x\in \bigcup_{X\in\mathcal{E}}X$, a (dummy or element) job corresponding to $x$ is eligible on each machine $e$ with $x\in e$.
\end{claim}
Using these claims, we can conclude the proof of Theorem \ref{thm:result_rar}:
\begin{lemma}
There is a perfect matching for the given \tailoredmatching instance, if and only if there is a schedule with makespan $47$ for the constructed \rar{2} instance. 
\end{lemma}
\begin{proof}
Let $F$ be a perfect matching for the \tailoredmatching instance.
For each $x\in \bigcup_{X\in\mathcal{E}}X$, we assign the corresponding element job to the machine $e$ with $x\in e$ and $e\in F$.
Furthermore, the dummy jobs corresponding to $x\in X$ with $X\in\mathcal{E}$, are distributed to the machines $e$ with $x\in e$ and $e\notin F$ such that each such machine receives exactly one job.
Hence, each machine $e\in E$ receives exactly three eligible jobs either with sizes $\alpha_A$, $\alpha_B$ and $\alpha_C$ or $\beta_A$, $\beta_B$ and $\beta_C$. 

Next, we assume that there is a schedule with makespan $47$ for the scheduling instance.
There are exactly $|\machs|$ many jobs with size $\alpha_A = \alpha_{A'}$ or $\beta_A = \beta_{A'}$ corresponding to elements of $A\cup A'$, and due to Claim \ref{claim:3res_3jobspermach} we know that each machine receives exactly one of these jobs.
The machines corresponding to triplets from $E(a_{3n})$ are the only ones that can process the $|E(a_{3n})|$ jobs corresponding to $a_{3n}$, and hence each of these machines receives exactly one of these jobs.
Now, the machines corresponding to triplets from $E(a'_{3n})$ are the only \emph{remaining} ones that can process the $|E(a'_{3n})|$ jobs corresponding to $a'_{3n}$.
Iterating this argument, we get that each machine $e$ receives exactly one job corresponding to some $x\in A\cup A'$ with $x\in e$.
Note that the above argument was based on the first resource value. 
Considering the second resource value yields the same result for each $x\in C \cup C'$.
For the elements $x\in B\cup B'$ both resource values have to be considered, namely the second for $b\in B$ and the first for $b'\in B'$, but the argument stays the same.
Summing up, each machine $e = \set{x,y,z}$ receives exactly three jobs corresponding to $x$, $y$ and $z$.
Now, considering this and Claim \ref{claim:3res_3jobspermach}, we get a perfect matching by selecting the triplets $e$ that processes three element jobs.
\end{proof}

\section{Conclusion}
\label{sec:conclusion}

In this paper we provided hardness of approximation results for scheduling with interval and resource restrictions.
We list some possible future research directions:

From the perspective of complexity, tighter hardness results seem plausible.
In particular, we have the same inapproximability results for \rar{2} and \rar{3} and it would be interesting to find a better result for \rar{3}.

From the algorithmic perspective, it remains open whether any of the studied problems and \rai in particular admits an approximation algorithm with a rate smaller than $2$.
There have been some results \cite{WangS16,Schwarz10Thesis} for \rai using promising linear programming relaxations that may be useful in this context.
Another possibility is the application of the local search techniques originally used by Svensson \cite{Svensson12} for the restricted assignment problem.
This approach recently yielded a breakthrough for the graph balancing problem \cite{JansenR18GrapbBalancing}.

Finally, while a PTAS for \rar{1} is known \cite{OLL08}, it is unclear whether the problem admits a so called \emph{efficient} PTAS with a running time of the form $f(1/\eps)\mathrm{poly}(|I|)$ for some computable function $f$.

\bibliography{references}

\begin{thebibliography}{10}

\bibitem{AnnamalaiKS17}
Chidambaram Annamalai, Christos Kalaitzis, and Ola Svensson.
\newblock Combinatorial algorithm for restricted max-min fair allocation.
\newblock {\em {ACM} Trans. Algorithms}, 13(3):37:1--37:28, 2017.
\newblock \href {http://dx.doi.org/10.1145/3070694}
  {\path{doi:10.1145/3070694}}.

\bibitem{BhaskaraKTW13}
Aditya Bhaskara, Ravishankar Krishnaswamy, Kunal Talwar, and Udi Wieder.
\newblock Minimum makespan scheduling with low rank processing times.
\newblock In {\em Proceedings of the Twenty-Fourth Annual {ACM-SIAM} Symposium
  on Discrete Algorithms, {SODA} 2013, New Orleans, Louisiana, USA, January
  6-8, 2013}, pages 937--947, 2013.
\newblock \href {http://dx.doi.org/10.1137/1.9781611973105.67}
  {\path{doi:10.1137/1.9781611973105.67}}.

\bibitem{BrandstadtL03}
Andreas Brandst{\"{a}}dt and Vadim~V. Lozin.
\newblock On the linear structure and clique-width of bipartite permutation
  graphs.
\newblock {\em Ars Comb.}, 67, 2003.

\bibitem{ChakrabartyKL15}
Deeparnab Chakrabarty, Sanjeev Khanna, and Shi Li.
\newblock On (1, \emph{{$\eps$}})-restricted assignment makespan minimization.
\newblock In {\em Proceedings of the Twenty-Sixth Annual {ACM-SIAM} Symposium
  on Discrete Algorithms, {SODA} 2015, San Diego, CA, USA, January 4-6, 2015},
  pages 1087--1101, 2015.
\newblock \href {http://dx.doi.org/10.1137/1.9781611973730.73}
  {\path{doi:10.1137/1.9781611973730.73}}.

\bibitem{ChakrabartyS16}
Deeparnab Chakrabarty and Kirankumar Shiragur.
\newblock Graph balancing with two edge types.
\newblock {\em CoRR}, abs/1604.06918, 2016.
\newblock \href {http://arxiv.org/abs/1604.06918} {\path{arXiv:1604.06918}}.

\bibitem{CMYZ17}
Lin Chen, D{\'{a}}niel Marx, Deshi Ye, and Guochuan Zhang.
\newblock Parameterized and approximation results for scheduling with a low
  rank processing time matrix.
\newblock In {\em 34th Symposium on Theoretical Aspects of Computer Science,
  {STACS} 2017, March 8-11, 2017, Hannover, Germany}, pages 22:1--22:14, 2017.
\newblock \href {http://dx.doi.org/10.4230/LIPIcs.STACS.2017.22}
  {\path{doi:10.4230/LIPIcs.STACS.2017.22}}.

\bibitem{ChenYZ14}
Lin Chen, Deshi Ye, and Guochuan Zhang.
\newblock An improved lower bound for rank four scheduling.
\newblock {\em Oper. Res. Lett.}, 42(5):348--350, 2014.
\newblock \href {http://dx.doi.org/10.1016/j.orl.2014.06.003}
  {\path{doi:10.1016/j.orl.2014.06.003}}.

\bibitem{ChengM18}
Siu{-}Wing Cheng and Yuchen Mao.
\newblock Restricted max-min fair allocation.
\newblock In {\em 45th International Colloquium on Automata, Languages, and
  Programming, {ICALP} 2018, July 9-13, 2018, Prague, Czech Republic}, pages
  37:1--37:13, 2018.
\newblock \href {http://dx.doi.org/10.4230/LIPIcs.ICALP.2018.37}
  {\path{doi:10.4230/LIPIcs.ICALP.2018.37}}.

\bibitem{DRZ18}
Sami Davies, Thomas Rothvoss, and Yihao Zhang.
\newblock A tale of santa claus, hypergraphs and matroids.
\newblock {\em CoRR}, abs/1807.07189, 2018.
\newblock \href {http://arxiv.org/abs/1807.07189} {\path{arXiv:1807.07189}}.

\bibitem{EbenlendrKS14}
Tom{\'{a}}s Ebenlendr, Marek Krc{\'{a}}l, and Jir{\'{\i}} Sgall.
\newblock Graph balancing: {A} special case of scheduling unrelated parallel
  machines.
\newblock {\em Algorithmica}, 68(1):62--80, 2014.
\newblock \href {http://dx.doi.org/10.1007/s00453-012-9668-9}
  {\path{doi:10.1007/s00453-012-9668-9}}.

\bibitem{EpsteinL11}
Leah Epstein and Asaf Levin.
\newblock Scheduling with processing set restrictions: Ptas results for several
  variants.
\newblock {\em International Journal of Production Economics}, 133(2):586--595,
  2011.
\newblock \href {http://dx.doi.org/10.1016/j.ijpe.2011.04.024}
  {\path{doi:10.1016/j.ijpe.2011.04.024}}.

\bibitem{HeggernesK07}
Pinar Heggernes and Dieter Kratsch.
\newblock Linear-time certifying recognition algorithms and forbidden induced
  subgraphs.
\newblock {\em Nord. J. Comput.}, 14(1-2):87--108, 2007.

\bibitem{HochbaumS87identische}
Dorit~S. Hochbaum and David~B. Shmoys.
\newblock Using dual approximation algorithms for scheduling problems
  theoretical and practical results.
\newblock {\em J. {ACM}}, 34(1):144--162, 1987.
\newblock \href {http://dx.doi.org/10.1145/7531.7535}
  {\path{doi:10.1145/7531.7535}}.

\bibitem{HochbaumS88uniforme}
Dorit~S. Hochbaum and David~B. Shmoys.
\newblock A polynomial approximation scheme for scheduling on uniform
  processors: Using the dual approximation approach.
\newblock {\em {SIAM} J. Comput.}, 17(3):539--551, 1988.
\newblock \href {http://dx.doi.org/10.1137/0217033}
  {\path{doi:10.1137/0217033}}.

\bibitem{HorowitzS76}
Ellis Horowitz and Sartaj Sahni.
\newblock Exact and approximate algorithms for scheduling nonidentical
  processors.
\newblock {\em J. {ACM}}, 23(2):317--327, 1976.
\newblock \href {http://dx.doi.org/10.1145/321941.321951}
  {\path{doi:10.1145/321941.321951}}.

\bibitem{HuangO16}
Chien{-}Chung Huang and Sebastian Ott.
\newblock A combinatorial approximation algorithm for graph balancing with
  light hyper edges.
\newblock In {\em 24th Annual European Symposium on Algorithms, {ESA} 2016,
  August 22-24, 2016, Aarhus, Denmark}, pages 49:1--49:15, 2016.
\newblock \href {http://dx.doi.org/10.4230/LIPIcs.ESA.2016.49}
  {\path{doi:10.4230/LIPIcs.ESA.2016.49}}.

\bibitem{JansenMS17}
Klaus Jansen, Marten Maack, and Roberto Solis{-}Oba.
\newblock Structural parameters for scheduling with assignment restrictions.
\newblock In {\em Algorithms and Complexity - 10th International Conference,
  {CIAC} 2017, Athens, Greece, May 24-26, 2017, Proceedings}, pages 357--368,
  2017.
\newblock \href {http://dx.doi.org/10.1007/978-3-319-57586-5\_30}
  {\path{doi:10.1007/978-3-319-57586-5\_30}}.

\bibitem{JansenR17SODA}
Klaus Jansen and Lars Rohwedder.
\newblock On the configuration-lp of the restricted assignment problem.
\newblock In {\em Proceedings of the Twenty-Eighth Annual {ACM-SIAM} Symposium
  on Discrete Algorithms, {SODA} 2017, Barcelona, Spain, Hotel Porta Fira,
  January 16-19}, pages 2670--2678, 2017.
\newblock \href {http://dx.doi.org/10.1137/1.9781611974782.176}
  {\path{doi:10.1137/1.9781611974782.176}}.

\bibitem{JansenR17IPCO}
Klaus Jansen and Lars Rohwedder.
\newblock A quasi-polynomial approximation for the restricted assignment
  problem.
\newblock In {\em Integer Programming and Combinatorial Optimization - 19th
  International Conference, {IPCO} 2017, Waterloo, ON, Canada, June 26-28,
  2017, Proceedings}, pages 305--316, 2017.
\newblock \href {http://dx.doi.org/10.1007/978-3-319-59250-3\_25}
  {\path{doi:10.1007/978-3-319-59250-3\_25}}.

\bibitem{JansenR18GrapbBalancing}
Klaus Jansen and Lars Rohwedder.
\newblock Local search breaks 1.75 for graph balancing.
\newblock {\em CoRR}, abs/1811.00955, 2018.
\newblock \href {http://arxiv.org/abs/1811.00955} {\path{arXiv:1811.00955}}.

\bibitem{Khodamoradi16Thesis}
Kamyar Khodamoradi.
\newblock {\em Algorithms for Scheduling and Routing Problems}.
\newblock PhD thesis, Simon Fraser University, 2016.

\bibitem{KhodamoradiKRS13}
Kamyar Khodamoradi, Ramesh Krishnamurti, Arash Rafiey, and Georgios Stamoulis.
\newblock {PTAS} for ordered instances of resource allocation problems.
\newblock In {\em {IARCS} Annual Conference on Foundations of Software
  Technology and Theoretical Computer Science, {FSTTCS} 2013, December 12-14,
  2013, Guwahati, India}, pages 461--473, 2013.
\newblock \href {http://dx.doi.org/10.4230/LIPIcs.FSTTCS.2013.461}
  {\path{doi:10.4230/LIPIcs.FSTTCS.2013.461}}.

\bibitem{KhodamoradiKRS16}
Kamyar Khodamoradi, Ramesh Krishnamurti, Arash Rafiey, and Georgios Stamoulis.
\newblock {PTAS} for ordered instances of resource allocation problems with
  restrictions on inclusions.
\newblock {\em CoRR}, abs/1610.00082, 2016.
\newblock \href {http://arxiv.org/abs/1610.00082} {\path{arXiv:1610.00082}}.

\bibitem{LeeLP13survey}
Kangbok Lee, Joseph~Y.{-}T. Leung, and Michael~L. Pinedo.
\newblock Makespan minimization in online scheduling with machine eligibility.
\newblock {\em Annals {OR}}, 204(1):189--222, 2013.
\newblock \href {http://dx.doi.org/10.1007/s10479-012-1271-6}
  {\path{doi:10.1007/s10479-012-1271-6}}.

\bibitem{LenstraST90}
Jan~Karel Lenstra, David~B. Shmoys, and {\'{E}}va Tardos.
\newblock Approximation algorithms for scheduling unrelated parallel machines.
\newblock {\em Math. Program.}, 46:259--271, 1990.
\newblock \href {http://dx.doi.org/10.1007/BF01585745}
  {\path{doi:10.1007/BF01585745}}.

\bibitem{LeungL08survey}
Joseph Y-T Leung and Chung-Lun Li.
\newblock Scheduling with processing set restrictions: A survey.
\newblock {\em International Journal of Production Economics}, 116(2):251--262,
  2008.
\newblock \href {http://dx.doi.org/10.1016/j.ijpe.2008.09.003}
  {\path{doi:10.1016/j.ijpe.2008.09.003}}.

\bibitem{LeungL16update}
Joseph Y-T Leung and Chung-Lun Li.
\newblock Scheduling with processing set restrictions: A literature update.
\newblock {\em International Journal of Production Economics}, 175:1--11, 2016.
\newblock \href {http://dx.doi.org/10.1016/j.ijpe.2014.09.038}
  {\path{doi:10.1016/j.ijpe.2014.09.038}}.

\bibitem{MuratoreSW10}
Gabriella Muratore, Ulrich~M. Schwarz, and Gerhard~J. Woeginger.
\newblock Parallel machine scheduling with nested job assignment restrictions.
\newblock {\em Oper. Res. Lett.}, 38(1):47--50, 2010.
\newblock \href {http://dx.doi.org/10.1016/j.orl.2009.09.010}
  {\path{doi:10.1016/j.orl.2009.09.010}}.

\bibitem{OLL08}
Jinwen Ou, Joseph Y-T Leung, and Chung-Lun Li.
\newblock Scheduling parallel machines with inclusive processing set
  restrictions.
\newblock {\em Naval Research Logistics (NRL)}, 55(4):328--338, 2008.
\newblock \href {http://dx.doi.org/10.1002/nav.20286}
  {\path{doi:10.1002/nav.20286}}.

\bibitem{PageS16}
Daniel~R. Page and Roberto Solis{-}Oba.
\newblock A 3/2-approximation algorithm for the graph balancing problem with
  two weights.
\newblock {\em Algorithms}, 9(2):38, 2016.
\newblock \href {http://dx.doi.org/10.3390/a9020038}
  {\path{doi:10.3390/a9020038}}.

\bibitem{Schaefer78}
Thomas~J. Schaefer.
\newblock The complexity of satisfiability problems.
\newblock In {\em Proceedings of the 10th Annual {ACM} Symposium on Theory of
  Computing, May 1-3, 1978, San Diego, California, {USA}}, pages 216--226,
  1978.
\newblock \href {http://dx.doi.org/10.1145/800133.804350}
  {\path{doi:10.1145/800133.804350}}.

\bibitem{SchuurmanW99}
Petra Schuurman and Gerhard~J Woeginger.
\newblock Polynomial time approximation algorithms for machine scheduling: Ten
  open problems.
\newblock {\em Journal of Scheduling}, 2(5):203--213, 1999.
\newblock \href
  {http://dx.doi.org/10.1002/(SICI)1099-1425(199909/10)2:5<203::AID-JOS26>3.0.CO;2-5}
  {\path{doi:10.1002/(SICI)1099-1425(199909/10)2:5<203::AID-JOS26>3.0.CO;2-5}}.

\bibitem{Schwarz10Thesis}
Ulrich~M. Schwarz.
\newblock {\em Approximation algorithms for scheduling and two-dimensional
  packing problems}.
\newblock PhD thesis, University of Kiel, 2010.
\newblock URL:
  \url{http://eldiss.uni-kiel.de/macau/receive/dissertation\_diss\_00005147}.

\bibitem{Schwarz10}
Ulrich~M. Schwarz.
\newblock A {PTAS} for scheduling with tree assignment restrictions.
\newblock {\em CoRR}, abs/1009.4529, 2010.
\newblock \href {http://arxiv.org/abs/1009.4529} {\path{arXiv:1009.4529}}.

\bibitem{PTASflawed}
Georgios Stamoulis.
\newblock Private communication, 2019.

\bibitem{Svensson12}
Ola Svensson.
\newblock Santa claus schedules jobs on unrelated machines.
\newblock {\em {SIAM} J. Comput.}, 41(5):1318--1341, 2012.
\newblock \href {http://dx.doi.org/10.1137/110851201}
  {\path{doi:10.1137/110851201}}.

\bibitem{Tovey84}
Craig~A. Tovey.
\newblock A simplified np-complete satisfiability problem.
\newblock {\em Discrete Applied Mathematics}, 8(1):85--89, 1984.
\newblock \href {http://dx.doi.org/10.1016/0166-218X(84)90081-7}
  {\path{doi:10.1016/0166-218X(84)90081-7}}.

\bibitem{WangS16}
Chao Wang and Ren{\'{e}} Sitters.
\newblock On some special cases of the restricted assignment problem.
\newblock {\em Inf. Process. Lett.}, 116(11):723--728, 2016.
\newblock \href {http://dx.doi.org/10.1016/j.ipl.2016.06.007}
  {\path{doi:10.1016/j.ipl.2016.06.007}}.

\bibitem{WS11book}
David~P. Williamson and David~B. Shmoys.
\newblock {\em The Design of Approximation Algorithms}.
\newblock Cambridge University Press, 2011.
\newblock URL:
  \url{http://www.cambridge.org/de/knowledge/isbn/item5759340/?site\_locale=de\_DE}.

\bibitem{Woeginger97}
Gerhard~J. Woeginger.
\newblock A polynomial-time approximation scheme for maximizing the minimum
  machine completion time.
\newblock {\em Oper. Res. Lett.}, 20(4):149--154, 1997.
\newblock \href {http://dx.doi.org/10.1016/S0167-6377(96)00055-7}
  {\path{doi:10.1016/S0167-6377(96)00055-7}}.

\end{thebibliography}
\newpage
\appendix

\section{Examples for Section \ref{sec:interval}}

\subparagraph{Example Simple Reduction.}

The following formula is an instance of \tailoredsat with minimal size:
\[(\neg x_2, \neg x_3, x_1)_1 \wedge (\neg x_1, x_2, \neg x_3)_1 \wedge (x_3,\neg x_2, x_1)_2 \wedge (\neg x_1, x_3, x_2)_2\]
Note that the formula is fulfilled if all the variables take the value $\true$.
The corresponding restricted assignment instance is depicted in \cref{fig:simple_example}.

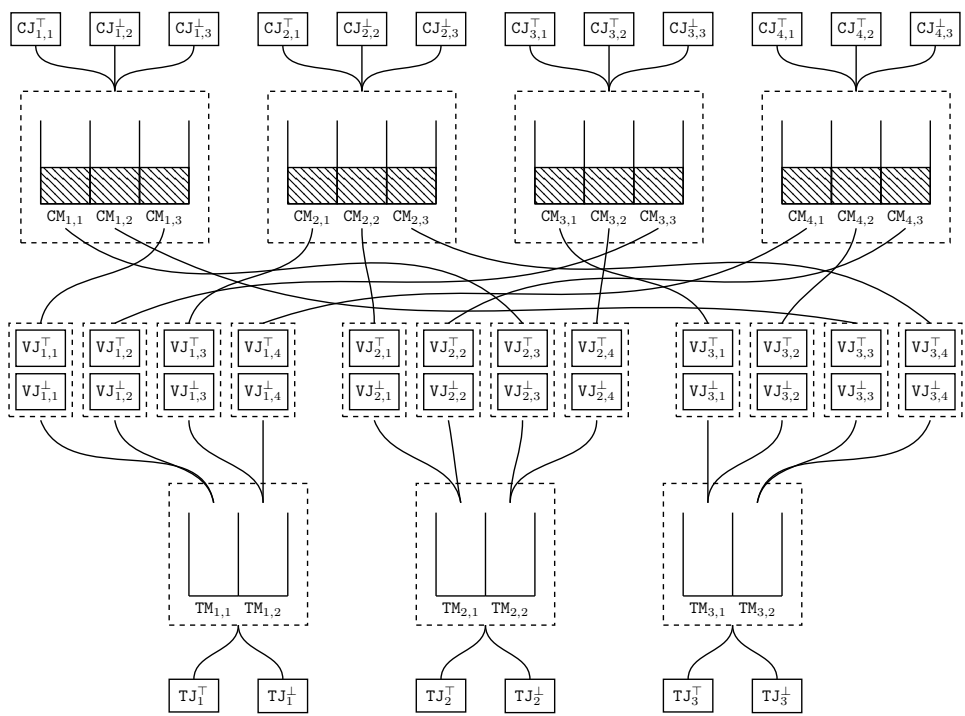
\begin{figure}[h]
\centering

\scalebox{0.65}{
\begin{tikzpicture}
\pgfmathsetmacro{\mw}{1.0}
\pgfmathsetmacro{\mh}{1.7}
\pgfmathsetmacro{\uw}{5.0}
\pgfmathsetmacro{\uh}{8.0}
\pgfmathsetmacro{\lw}{1.0}
\pgfmathsetmacro{\ls}{3.0*\mw}
\pgfmathsetmacro{\pj}{1.11/1.5}
\pgfmathsetmacro{\cjt}{1.00/1.5}
\pgfmathsetmacro{\cjf}{1.01/1.5}
\pgfmathsetmacro{\tjt}{1.00/1.5}
\pgfmathsetmacro{\tjf}{1.02/1.5}
\pgfmathsetmacro{\vjt}{1.11/1.5}
\pgfmathsetmacro{\vjf}{1.10/1.5}

\foreach \x in {0,...,3}{ 
    \tikzmath{integer \z; \z = \x + 1;} 
    \draw[thick,dashed] (\x * \uw - 0.4*\mw, \uh - 0.35*\mh - 0.2) rectangle (\x * \uw + 0.4*\mw + 3*\mw, \uh + 0.35*\mh + \mh);
    \foreach \y in {0,...,3}{
        \draw[thick] (\x * \uw + \y*\mw, \uh) -- (\x * \uw + \y*\mw, \uh + \mh);
    }
    \foreach \y in {1,...,3}{
        \draw[thick] (\x * \uw + \y*\mw - \mw, \uh) -- (\x * \uw+ \y*\mw, \uh) node [midway, yshift = - 0.3 cm] {$\mathtt{CM}_{\z,\y}$};
        \draw[thick, pattern=north west lines] (\x * \uw + \y*\mw - \mw, \uh) rectangle (\x * \uw+ \y*\mw, \uh + \pj);
    }
    \draw[thick] (\x * \uw - 0.6*\mw, \uh + 1.9*\mh) rectangle (\x * \uw + 0.4*\mw, \uh + 1.9*\mh + \cjt) node [midway] {$\mathtt{CJ}_{\z,1}^{\true}$};
    \draw[thick] (\x * \uw - 0.1*\mw, \uh + 1.9*\mh) to [out=270,in=90] (\x * \uw + 1.5*\mw, \uh + 0.35*\mh + \mh);
    \draw[thick] (\x * \uw + 2.6*\mw, \uh + 1.9*\mh) rectangle (\x * \uw + 3.6*\mw, \uh + 1.9*\mh + \cjf) node [midway] {$\mathtt{CJ}_{\z,3}^{\false}$};
    \draw[thick] (\x * \uw + 3.1*\mw, \uh + 1.9*\mh) to [out=270,in=90] (\x * \uw + 1.5*\mw, \uh + 0.35*\mh + \mh);
    \draw[thick] (\x * \uw + 1.5*\mw, \uh + 1.9*\mh) to [out=270,in=90] (\x * \uw + 1.5*\mw, \uh + 0.35*\mh + \mh);
}
\draw[thick] (0 * \uw + 1.0*\mw, \uh + 1.9*\mh) rectangle (0 * \uw + 2.0*\mw, \uh + 1.9*\mh + \cjf) node [midway] {$\mathtt{CJ}_{1,2}^{\false}$};
\draw[thick] (1 * \uw + 1.0*\mw, \uh + 1.9*\mh) rectangle (1 * \uw + 2.0*\mw, \uh + 1.9*\mh + \cjf) node [midway] {$\mathtt{CJ}_{2,2}^{\false}$};
\draw[thick] (2 * \uw + 1.0*\mw, \uh + 1.9*\mh) rectangle (2 * \uw + 2.0*\mw, \uh + 1.9*\mh + \cjt) node [midway] {$\mathtt{CJ}_{3,2}^{\true}$};
\draw[thick] (3 * \uw + 1.0*\mw, \uh + 1.9*\mh) rectangle (3 * \uw + 2.0*\mw, \uh + 1.9*\mh + \cjt) node [midway] {$\mathtt{CJ}_{4,2}^{\true}$};

\foreach \x in {0,...,2}{
    \tikzmath{integer \z; \z = \x + 1;} 
    \draw[thick,dashed] (\ls + \x * \uw - 0.4*\mw, - 0.35*\mh) rectangle (\ls + \x * \uw + 0.4*\mw + 2*\mw,  0.35*\mh + \mh);
    \foreach \y in {0,...,2}{
        \draw[thick] (\ls + \x * \uw + \y*\mw, 0) -- (\ls + \x * \uw + \y*\mw, \mh);
    }
    \foreach \y in {1,...,2}{
        \draw[thick] (\ls + \x * \uw + \y*\mw - \mw, 0) -- (\ls + \x * \uw+ \y*\mw, 0) node [midway, yshift = - 0.3 cm] {$\mathtt{TM}_{\z,\y}$};
    }
    \draw[thick] (\ls + \x * \uw - 0.4*\mw, -1.0*\mh) rectangle (\ls + \x * \uw + 0.6*\mw, -1.0*\mh - \tjt) node [midway] {$\mathtt{TJ}_{\z}^{\true}$};
    \draw[thick] (\ls + \x * \uw + 0.1*\mw, -1.0*\mh) to [out=90,in=270] (\ls + \x * \uw + \mw, - 0.35*\mh);
    \draw[thick] (\ls + \x * \uw + 1.4*\mw, -1.0*\mh) rectangle (\ls + \x * \uw + 2.4*\mw, -1.0*\mh - \tjf) node [midway] {$\mathtt{TJ}_{\z}^{\false}$};
    \draw[thick] (\ls + \x * \uw + 1.9*\mw, -1.0*\mh) to [out=90,in=270] (\ls + \x * \uw + \mw, - 0.35*\mh);
    \tikzmath{\s = 0.35;} 
    \draw[thick] (\ls + \x * \uw - 1.75*\mw + 0.35*\x*\uw - 0.35*\uw , 0.53*\uh + 0.6 * \vjt) rectangle (\ls + \x * \uw + \mw - 1.75*\mw+ 0.35*\x*\uw - 0.35*\uw, 0.53*\uh + \vjt + 0.6 * \vjt) node [midway] {$\mathtt{VJ}_{\z,1}^{\true}$};
    \draw[thick] (\ls + \x * \uw - 1.75*\mw+ 0.35*\x*\uw - 0.35*\uw, 0.53*\uh - 0.6 * \vjf) rectangle (\ls + \x * \uw + \mw - 1.75*\mw+ 0.35*\x*\uw - 0.35*\uw, 0.53*\uh + \vjf - 0.6 * \vjf) node [midway] {$\mathtt{VJ}_{\z,1}^{\false}$};
    \draw[thick,dashed] (\ls + \x * \uw - 1.75*\mw+ 0.35*\x*\uw - 0.35*\uw - 0.14, 0.53*\uh - 0.6 * \vjf - 0.14) rectangle (\ls + \x * \uw + \mw - 1.75*\mw+ 0.35*\x*\uw - 0.35*\uw + 0.14 , 0.53*\uh + \vjt + 0.6 * \vjt + 0.14);

    \draw[thick] (\ls + \x * \uw - 1.75*\mw + 1.5*\mw + 0.35*\x*\uw - 0.35*\uw, 0.53*\uh + 0.6 * \vjt) rectangle (\ls + \x * \uw + \mw - 1.75*\mw+ 1.5*\mw+ 0.35*\x*\uw - 0.35*\uw, 0.53*\uh + \vjt + 0.6 * \vjt) node [midway] {$\mathtt{VJ}_{\z,2}^{\true}$};
    \draw[thick] (\ls + \x * \uw - 1.75*\mw+ 1.5*\mw+ 0.35*\x*\uw - 0.35*\uw, 0.53*\uh - 0.6 * \vjf) rectangle (\ls + \x * \uw + \mw - 1.75*\mw+ 1.5*\mw + 0.35*\x*\uw - 0.35*\uw, 0.53*\uh + \vjf - 0.6 * \vjf) node [midway] {$\mathtt{VJ}_{\z,2}^{\false}$};
    \draw[thick,dashed] (\ls + \x * \uw - 1.75*\mw+ 1.5*\mw+ 0.35*\x*\uw - 0.35*\uw - 0.14, 0.53*\uh - 0.6 * \vjf - 0.14) rectangle (\ls + \x * \uw + \mw - 1.75*\mw+ 1.5*\mw+ 0.35*\x*\uw - 0.35*\uw + 0.14, 0.53*\uh + \vjt + 0.6 * \vjt + 0.14);

    \draw[thick] (\ls + \x * \uw - 1.75*\mw + 3*\mw + 0.35*\x*\uw - 0.35*\uw , 0.53*\uh + 0.6 * \vjt) rectangle (\ls + \x * \uw + \mw - 1.75*\mw+ 3*\mw + 0.35*\x*\uw - 0.35*\uw , 0.53*\uh + \vjt + 0.6 * \vjt) node [midway] {$\mathtt{VJ}_{\z,3}^{\true}$};
    \draw[thick] (\ls + \x * \uw - 1.75*\mw+ 3*\mw + 0.35*\x*\uw - 0.35*\uw, 0.53*\uh - 0.6 * \vjf) rectangle (\ls + \x * \uw + \mw - 1.75*\mw+ 3*\mw + 0.35*\x*\uw - 0.35*\uw, 0.53*\uh + \vjf - 0.6 * \vjf) node [midway] {$\mathtt{VJ}_{\z,3}^{\false}$};
    \draw[thick,dashed] (\ls + \x * \uw - 1.75*\mw+ 3*\mw + 0.35*\x*\uw - 0.35*\uw- 0.14, 0.53*\uh - 0.6 * \vjf- 0.14) rectangle (\ls + \x * \uw + \mw - 1.75*\mw+ 3*\mw + 0.35*\x*\uw - 0.35*\uw+ 0.14 , 0.53*\uh + \vjt + 0.6 * \vjt + 0.14);

    \draw[thick] (\ls + \x * \uw - 1.75*\mw + 4.5*\mw + 0.35*\x*\uw - 0.35*\uw , 0.53*\uh + 0.6 * \vjt) rectangle (\ls + \x * \uw + \mw - 1.75*\mw+ 4.5*\mw + 0.35*\x*\uw - 0.35*\uw , 0.53*\uh + \vjt + 0.6 * \vjt) node [midway] {$\mathtt{VJ}_{\z,4}^{\true}$};
    \draw[thick] (\ls + \x * \uw - 1.75*\mw+ 4.5*\mw + 0.35*\x*\uw - 0.35*\uw , 0.53*\uh - 0.6 * \vjf) rectangle (\ls + \x * \uw + \mw - 1.75*\mw+ 4.5*\mw + 0.35*\x*\uw - 0.35*\uw , 0.53*\uh + \vjf - 0.6 * \vjf) node [midway] {$\mathtt{VJ}_{\z,4}^{\false}$};
    \draw[thick,dashed] (\ls + \x * \uw - 1.75*\mw+ 4.5*\mw + 0.35*\x*\uw - 0.35*\uw- 0.14 , 0.53*\uh - 0.6 * \vjf- 0.14) rectangle (\ls + \x * \uw + \mw - 1.75*\mw+ 4.5*\mw + 0.35*\x*\uw - 0.35*\uw+ 0.14 , 0.53*\uh + \vjt + 0.6 * \vjt+ 0.14);
}

\draw[thick] (\ls + 0 * \uw + 0.5*\mw - 1.75*\mw+ 0.35*0*\uw - 0.35*\uw , 0.53*\uh + \vjt + 0.6 * \vjt + 0.14) to [out=80,in=260] (2.5*\mw +0*\uw ,\uh - 0.5);
\draw[thick] (\ls + 0 * \uw + 0.5*\mw - 1.75*\mw+ 0.35*0*\uw - 0.35*\uw , 0.53*\uh + \vjt + 0.6 * \vjt - 0.14 -\mh) to [out=290,in=100] (\ls + 0.5*\mw ,\mh + 0.2);
\draw[thick] (\ls + 0 * \uw + 0.5*\mw - 1.75*\mw+ 0.35*0*\uw - 0.35*\uw + 1.5*\mw, 0.53*\uh + \vjt + 0.6 * \vjt + 0.14) to [out=25,in=210] (2.5*\mw +2*\uw ,\uh - 0.5);
\draw[thick] (\ls + 0 * \uw + 0.5*\mw - 1.75*\mw+ 0.35*0*\uw - 0.35*\uw + 1.5*\mw, 0.53*\uh + \vjt + 0.6 * \vjt - 0.14 -\mh) to [out=280,in=100] (\ls + 0.5*\mw ,\mh + 0.2);
\draw[thick] (\ls + 0 * \uw + 0.5*\mw - 1.75*\mw+ 0.35*0*\uw - 0.35*\uw + 3.0*\mw, 0.53*\uh + \vjt + 0.6 * \vjt + 0.14) to [out=75,in=265] (0.5*\mw +1*\uw ,\uh - 0.5);
\draw[thick] (\ls + 0 * \uw + 0.5*\mw - 1.75*\mw+ 0.35*0*\uw - 0.35*\uw + 3.0*\mw, 0.53*\uh + \vjt + 0.6 * \vjt - 0.14 -\mh) to [out=270,in=90] (\ls + 1.5*\mw ,\mh + 0.2);
\draw[thick] (\ls + 0 * \uw + 0.5*\mw - 1.75*\mw+ 0.35*0*\uw - 0.35*\uw + 4.5*\mw, 0.53*\uh + \vjt + 0.6 * \vjt + 0.14) to [out=20,in=215] (0.5*\mw +3*\uw ,\uh - 0.5);
\draw[thick] (\ls + 0 * \uw + 0.5*\mw - 1.75*\mw+ 0.35*0*\uw - 0.35*\uw + 4.5*\mw, 0.53*\uh + \vjt + 0.6 * \vjt - 0.14 -\mh) to [out=270,in=90] (\ls + 1.5*\mw ,\mh + 0.2);

\draw[thick] (\ls + 1 * \uw + 0.5*\mw - 1.75*\mw+ 0.35*1*\uw - 0.35*\uw , 0.53*\uh + \vjt + 0.6 * \vjt + 0.14) to [out=90,in=270] (1.5*\mw +1*\uw ,\uh - 0.5);
\draw[thick] (\ls + 1 * \uw + 0.5*\mw - 1.75*\mw+ 0.35*1*\uw - 0.35*\uw , 0.53*\uh + \vjt + 0.6 * \vjt - 0.14 -\mh) to [out=290,in=100] (\ls + 0.5*\mw + 1*\uw,\mh + 0.2);
\draw[thick] (\ls + 1 * \uw + 0.5*\mw - 1.75*\mw+ 0.35*1*\uw - 0.35*\uw + 1.5*\mw, 0.53*\uh + \vjt + 0.6 * \vjt + 0.14) to [out=35,in=220] (2.5*\mw +3*\uw ,\uh - 0.5);
\draw[thick] (\ls + 1 * \uw + 0.5*\mw - 1.75*\mw+ 0.35*1*\uw - 0.35*\uw + 1.5*\mw, 0.53*\uh + \vjt + 0.6 * \vjt - 0.14 -\mh) to [out=280,in=100] (\ls + 0.5*\mw + 1*\uw,\mh + 0.2);
\draw[thick] (\ls + 1 * \uw + 0.5*\mw - 1.75*\mw+ 0.35*1*\uw - 0.35*\uw + 3.0*\mw, 0.53*\uh + \vjt + 0.6 * \vjt + 0.14) to [out=130,in=320] (0.5*\mw +0*\uw ,\uh - 0.5);
\draw[thick] (\ls + 1 * \uw + 0.5*\mw - 1.75*\mw+ 0.35*1*\uw - 0.35*\uw + 3.0*\mw, 0.53*\uh + \vjt + 0.6 * \vjt - 0.14 -\mh) to [out=270,in=90] (\ls + 1.5*\mw + 1*\uw,\mh + 0.2);
\draw[thick] (\ls + 1 * \uw + 0.5*\mw - 1.75*\mw+ 0.35*1*\uw - 0.35*\uw + 4.5*\mw, 0.53*\uh + \vjt + 0.6 * \vjt + 0.14) to [out=85,in=265] (1.5*\mw +2*\uw ,\uh - 0.5);
\draw[thick] (\ls + 1 * \uw + 0.5*\mw - 1.75*\mw+ 0.35*1*\uw - 0.35*\uw + 4.5*\mw, 0.53*\uh + \vjt + 0.6 * \vjt - 0.14 -\mh) to [out=270,in=90] (\ls + 1.5*\mw + 1*\uw,\mh + 0.2);

\draw[thick] (\ls + 2 * \uw + 0.5*\mw - 1.75*\mw+ 0.35*2*\uw - 0.35*\uw , 0.53*\uh + \vjt + 0.6 * \vjt + 0.14) to [out=115,in=285] (0.5*\mw +2*\uw ,\uh - 0.5);
\draw[thick] (\ls + 2 * \uw + 0.5*\mw - 1.75*\mw+ 0.35*2*\uw - 0.35*\uw , 0.53*\uh + \vjt + 0.6 * \vjt - 0.14 -\mh) to [out=270,in=90] (\ls + 0.5*\mw + 2*\uw,\mh + 0.2);
\draw[thick] (\ls + 2 * \uw + 0.5*\mw - 1.75*\mw+ 0.35*2*\uw - 0.35*\uw + 1.5*\mw, 0.53*\uh + \vjt + 0.6 * \vjt + 0.14) to [out=60,in=260] (1.5*\mw +3*\uw ,\uh - 0.5);
\draw[thick] (\ls + 2 * \uw + 0.5*\mw - 1.75*\mw+ 0.35*2*\uw - 0.35*\uw + 1.5*\mw, 0.53*\uh + \vjt + 0.6 * \vjt - 0.14 -\mh) to [out=270,in=90] (\ls + 0.5*\mw + 2*\uw,\mh + 0.2);
\draw[thick] (\ls + 2 * \uw + 0.5*\mw - 1.75*\mw+ 0.35*2*\uw - 0.35*\uw + 3.0*\mw, 0.53*\uh + \vjt + 0.6 * \vjt + 0.14) to [out=170,in=330] (1.5*\mw +0*\uw ,\uh - 0.5);
\draw[thick] (\ls + 2 * \uw + 0.5*\mw - 1.75*\mw+ 0.35*2*\uw - 0.35*\uw + 3.0*\mw, 0.53*\uh + \vjt + 0.6 * \vjt - 0.14 -\mh) to [out=260,in=80] (\ls + 1.5*\mw + 2*\uw,\mh + 0.2);
\draw[thick] (\ls + 2 * \uw + 0.5*\mw - 1.75*\mw+ 0.35*2*\uw - 0.35*\uw + 4.5*\mw, 0.53*\uh + \vjt + 0.6 * \vjt + 0.14) to [out=140,in=330] (2.5*\mw +1*\uw ,\uh - 0.5);
\draw[thick] (\ls + 2 * \uw + 0.5*\mw - 1.75*\mw+ 0.35*2*\uw - 0.35*\uw + 4.5*\mw, 0.53*\uh + \vjt + 0.6 * \vjt - 0.14 -\mh) to [out=250,in=80] (\ls + 1.5*\mw + 2*\uw,\mh + 0.2);

\end{tikzpicture}
}
\caption{The restricted assignment instance constructed for a minimal example instance. The hatched rectangles represent private loads, and the connecting lines indicate eligibility. If these lines end at a dashed rectangle, the eligibility information concerns everything within the rectangle. We chose a short notation for the jobs and machines writing, e.g., $\mathtt{VJ}_{j,t}^\circ$ instead of $\variablejob_{j,t}^\circ$.}
\label{fig:simple_example}
\end{figure}

\subparagraph{Example Refined Reduction.}

We consider the same formula as above for the refined reduction.
The values of $\kappa$ for the occurrences of the first two variables together with the resulting increasing lexicographical ordering is depicted in \cref{table:example_refinded}.
Furthermore, in \cref{fig:interval_example} the truth assignment as well as the bridge and highway gadget for the first two variables are depicted.
\begin{table}[h]
\centering
\caption{The occurrences of the first two values in the clauses and the resulting increasing lexicographical ordering of the occurrences.}
\begin{tabular}{r|llllllll}
$(j,t)$ & (1,1) & (1,2) & (1,3) & (1,4) & (2,1) & (2,2) & (2,3) & (2,4)\\\hline
$\kappa(j,t)$ & (1,3) & (3,3) & (2,1) & (4,1) & (2,2) & (4,3) & (1,1) & (3,2)\\\hline
Ordering $(j,t)$ & (2,3) & (1,1) & (1,3) & (2,1) & (2,4) & (1,2) & (1,4) & (2,2)
\end{tabular}
\label{table:example_refinded}
\end{table}

\begin{sidewaysfigure}
\centering

\scalebox{0.75}{
\begin{tikzpicture}[scale = 1.0]
\pgfmathsetmacro{\mw}{1.3}
\pgfmathsetmacro{\bbh}{1.0}
\pgfmathsetmacro{\bsh}{0.5}
\pgfmathsetmacro{\ih}{0.5}
\pgfmathsetmacro{\ld}{-0.27}

\foreach \x in {0, 2, 6, 10, 12, 20}{
    \draw[very thick] (\x*\mw,\bbh) -- (\x*\mw,0);
}

\foreach \x in {0,...,19}{
    \draw[very thick] (\x*\mw ,0) -- (\x*\mw + \mw,0);
    \draw[very thick] (\x*\mw + \mw,\bsh) -- (\x*\mw + \mw,0);
}
\draw[very thick, dashed] (20*\mw ,0) -- (21*\mw,0);

\draw[decorate,decoration={brace,amplitude=8pt,raise=1pt,mirror},yshift=-0pt] (0*\mw,0) -- (2*\mw,0) node [midway,yshift=-18pt]{$\tblock_{1}$};
\draw[decorate,decoration={brace,amplitude=8pt,raise=1pt,mirror},yshift=-0pt] (2*\mw,0) -- (6*\mw,0) node [midway,yshift=-18pt]{$\sblock_{1}$};
\draw[decorate,decoration={brace,amplitude=8pt,raise=1pt,mirror},yshift=-0pt] (6*\mw,0) -- (10*\mw,0) node [midway,yshift=-18pt]{$\pblock_{2}$};
\draw[decorate,decoration={brace,amplitude=8pt,raise=1pt,mirror},yshift=-0pt] (10*\mw,0) -- (12*\mw,0) node [midway,yshift=-18pt]{$\tblock_{2}$};
\draw[decorate,decoration={brace,amplitude=8pt,raise=1pt,mirror},yshift=-0pt] (12*\mw,0) -- (20*\mw,0) node [midway,yshift=-18pt]{$\sblock_{2}$};


\foreach \x/\y in  {0.5/{$\mathtt{TM}_{1,1}$}, 
                    1.5/{$\mathtt{TM}_{1,2}$}, 
                    2.5/{$\mathtt{GM}_{1,4}$},
                    3.5/{$\mathtt{GM}_{1,2}$},
                    4.5/{$\mathtt{GM}_{1,3}$},
                    5.5/{$\mathtt{GM}_{1,1}$},
                    6.5/{$\mathtt{BMI}_{1,1,2}$},
                    7.5/{$\mathtt{BMI}_{1,3,2}$},
                    8.5/{$\mathtt{BMI}_{1,2,2}$},
                    9.5/{$\mathtt{BMI}_{1,4,2}$},
                    10.5/{$\mathtt{TM}_{2,1}$},
                    11.5/{$\mathtt{TM}_{2,2}$},
                    12.5/{$\mathtt{GM}_{2,2}$},
                    13.5/{$\mathtt{BMO}_{1,4,2}$},
                    14.5/{$\mathtt{BMO}_{1,2,2}$},
                    15.5/{$\mathtt{GM}_{2,4}$},
                    16.5/{$\mathtt{GM}_{2,1}$},
                    17.5/{$\mathtt{BMO}_{1,3,2}$},
                    18.5/{$\mathtt{BMO}_{1,1,2}$},
                    19.5/{$\mathtt{GM}_{2,3}$}
                    }{

\node at (\x*\mw, 0.2) {\y};
}


\draw[very thick, color = mlightblue] (0*\mw,\bbh + 1*\ih) -- (2*\mw - 0.07,\bbh + 1*\ih) node [midway, yshift = \ld cm, text = black] { $\truthjob_{1}^{\circ}$};
\draw[very thick, color = mlightblue] (10*\mw,\bbh + 1*\ih) -- (12*\mw - 0.07,\bbh + 1*\ih) node [midway, yshift = \ld cm, text = black] { $\truthjob_{2}^{\circ}$};


\draw[very thick, color = morange] (0*\mw,\bbh + 3*\ih) -- (6*\mw - 0.07,\bbh + 3*\ih) node [midway, yshift = \ld cm, text = black] {$\variablejob_{1,1}^{\circ}$};
\draw[very thick, color = morange] (1*\mw,\bbh + 4*\ih) -- (5*\mw - 0.07,\bbh + 4*\ih) node [midway, yshift = \ld cm, text = black] {$\variablejob_{1,3}^{\circ}$};
\draw[very thick, color = morange] (0*\mw,\bbh + 7*\ih) -- (4*\mw - 0.07,\bbh + 7*\ih) node [midway, yshift = \ld cm, text = black] {$\variablejob_{1,2}^{\circ}$};
\draw[very thick, color = morange] (1*\mw,\bbh + 8*\ih) -- (3*\mw - 0.07,\bbh + 8*\ih) node [midway, yshift = \ld cm, text = black] {$\variablejob_{1,4}^{\circ}$};


\draw[very thick, color = mdarkblue] (5*\mw + 0.07,\bbh + 4*\ih) -- (7*\mw - 0.07,\bbh + 4*\ih) node [midway, yshift = \ld cm, text = black] {$\highwayjob_{1,1,1}^{\circ}$};
\draw[very thick, color = mdarkblue] (4*\mw + 0.07,\bbh + 5*\ih) -- (8*\mw - 0.07,\bbh + 5*\ih) node [midway, yshift = \ld cm, text = black] {$\highwayjob_{1,3,1}^{\circ}$};
\draw[very thick, color = mdarkblue] (3*\mw + 0.07,\bbh + 8*\ih) -- (9*\mw - 0.07,\bbh + 8*\ih) node [midway, yshift = \ld cm, text = black] { $\highwayjob_{1,2,1}^{\circ}$};
\draw[very thick, color = mdarkblue] (2*\mw + 0.07,\bbh + 9*\ih) -- (10*\mw - 0.07,\bbh + 9*\ih) node [midway, yshift = \ld cm, text = black] { $\highwayjob_{1,4,1}^{\circ}$};


\draw[very thick, color = morange] (11*\mw + 0.07,\bbh + 2*\ih) -- (20*\mw - 0.07,\bbh + 2*\ih) node [midway, yshift = \ld cm, text = black] { $\variablejob_{2,3}^{\circ}$};
\draw[very thick, color = mred] (6*\mw + 0.07,\bbh + 3*\ih) -- (19*\mw - 0.07,\bbh + 3*\ih) node [midway, yshift = \ld cm, text = black] { $\bridgejob_{1,1,2}^{\circ}$};
\draw[very thick, color = mred] (7*\mw + 0.07,\bbh + 4*\ih) -- (18*\mw - 0.07,\bbh + 4*\ih) node [midway, yshift = \ld cm, text = black] { $\bridgejob_{1,3,2}^{\circ}$};
\draw[very thick, color = morange] (10*\mw + 0.07,\bbh + 5*\ih) -- (17*\mw - 0.07,\bbh + 5*\ih) node [midway, yshift = \ld cm, text = black] { $\variablejob_{2,1}^{\circ}$};
\draw[very thick, color = morange] (11*\mw + 0.07,\bbh + 6*\ih) -- (16*\mw - 0.07,\bbh + 6*\ih) node [midway, yshift = \ld cm, text = black] { $\variablejob_{2,4}^{\circ}$};
\draw[very thick, color = mred] (8*\mw + 0.07,\bbh + 7*\ih) -- (15*\mw - 0.07,\bbh + 7*\ih) node [midway, yshift = \ld cm, text = black] { $\bridgejob_{1,2,2}^{\circ}$};
\draw[very thick, color = mred] (9*\mw + 0.07,\bbh + 8*\ih) -- (14*\mw - 0.07,\bbh + 8*\ih) node [midway, yshift = \ld cm, text = black] { $\bridgejob_{1,4,2}^{\circ}$};
\draw[very thick, color = morange] (10*\mw + 0.07,\bbh + 9*\ih) -- (13*\mw - 0.07,\bbh + 9*\ih) node [midway, yshift = \ld cm, text = black] { $\variablejob_{2,2}^{\circ}$};


\draw[very thick, color = mdarkblue] (19*\mw + 0.07,\bbh + 3*\ih) -- (20*\mw - 0.07,\bbh + 3*\ih) node [at end, yshift = \ld cm, xshift = 0.1cm , text = black] { $\highwayjob_{2,3,2}^{\circ}$};
\draw[very thick, color = mdarkblue] (18*\mw + 0.07,\bbh + 4*\ih) -- (20*\mw - 0.07,\bbh + 4*\ih) node [at end, yshift = \ld cm, xshift = 0.1cm , text = black] { $\highwayjob_{1,1,2}^{\circ}$};
\draw[very thick, color = mdarkblue] (17*\mw + 0.07,\bbh + 5*\ih) -- (20*\mw - 0.07,\bbh + 5*\ih) node [at end, yshift = \ld cm, xshift = 0.1cm , text = black] { $\highwayjob_{1,3,2}^{\circ}$};
\draw[very thick, color = mdarkblue] (16*\mw + 0.07,\bbh + 6*\ih) -- (20*\mw - 0.07,\bbh + 6*\ih) node [at end, yshift = \ld cm, xshift = 0.1cm , text = black] { $\highwayjob_{2,1,2}^{\circ}$};
\draw[very thick, color = mdarkblue] (15*\mw + 0.07,\bbh + 7*\ih) -- (20*\mw - 0.07,\bbh + 7*\ih) node [at end, yshift = \ld cm, xshift = 0.1cm , text = black] { $\highwayjob_{2,4,2}^{\circ}$};
\draw[very thick, color = mdarkblue] (14*\mw + 0.07,\bbh + 8*\ih) -- (20*\mw - 0.07,\bbh + 8*\ih) node [at end, yshift = \ld cm, xshift = 0.1cm , text = black] { $\highwayjob_{1,2,2}^{\circ}$};
\draw[very thick, color = mdarkblue] (13*\mw + 0.07,\bbh + 9*\ih) -- (20*\mw - 0.07,\bbh + 9*\ih) node [at end, yshift = \ld cm, xshift = 0.1cm , text = black] { $\highwayjob_{1,4,2}^{\circ}$};
\draw[very thick, color = mdarkblue] (12*\mw + 0.07,\bbh + 10*\ih) -- (20*\mw - 0.07,\bbh + 10*\ih) node [at end, yshift = \ld cm, xshift = 0.1cm , text = black] { $\highwayjob_{2,2,2}^{\circ}$};

\foreach \x in {1,...,8}{
    \draw[very thick,dashed, color = mdarkblue] (20*\mw - 0.07,\bbh + 2*\ih + \x*\ih) -- (21*\mw,\bbh + 2*\ih + \x*\ih);
}
\end{tikzpicture}
}
\caption{The truth assignment as well as the bridge and highway gadget for the first two variables of an example instance. The colored lines mark the intervals of eligible machines for the respective jobs. In the picture, we use a more compact notation for the machines, and write, e.g., $\mathtt{TM}_{1,1}$ instead of $\truthmach_{1,1}$.}
\label{fig:interval_example}
\end{sidewaysfigure}
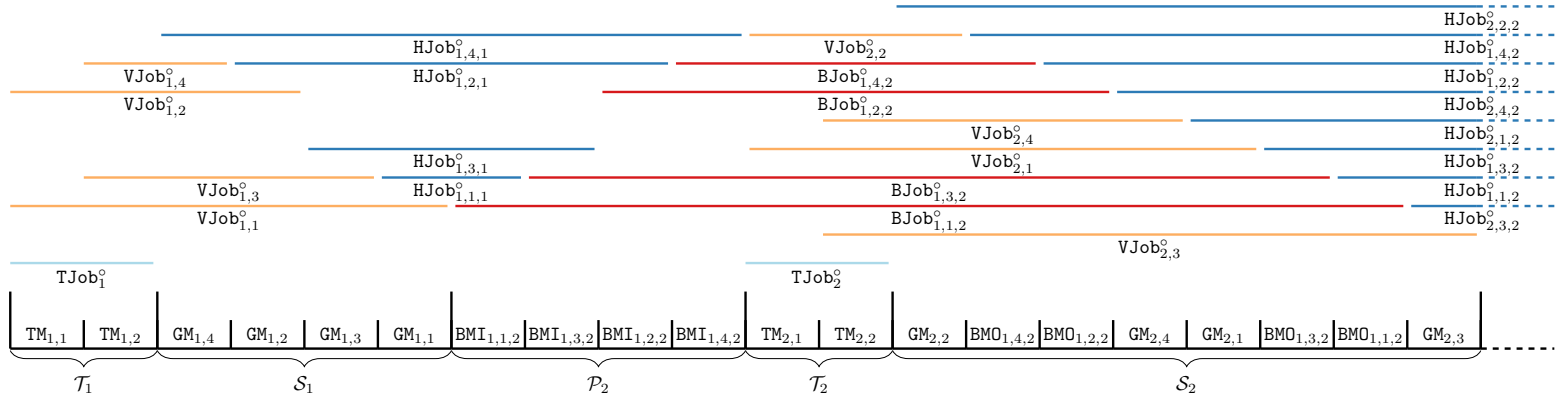

\section{Reduction by Bhaskara et al. \cite{BhaskaraKTW13}}\label{appendix_reduction}

We state the reduction by Bhaskara et al. \cite{BhaskaraKTW13} from \threedm to \lrs{4} and show that it is not sound.
We remark that the construction can be repaired with not too much effort using processing times similar to the ones presented in the reduction for \rar{3} in \cref{sec:resource}.  

Given an instance $(A,B,C,E)$ of \threedm and some $\eps>0$, let $n=|A|$, $N=n/\eps$ and $E(x) = \sett{e\in E}{x \in e}$ for each $x\in A\cup B\cup C$.
We identify the set of machines $\machs$ with the set of triplets $E$, i.e., $\machs = E$.
The speed vector of $e = \set{a_i,b_j,c_k} \in \machs$ is given by $(N^i,N^j,N^k,1)$.
Furthermore, for each $x\in A \cup B \cup C$, we introduce one element and $|E(x)| - 1 $ dummy jobs.
The size vectors of the jobs are presented in the following table:
\begin{center}
\begin{tabular}{l|ll}
 & Element & Dummy \\\midrule
$a_i\in A$ & $(\eps N^{-i},0,0,1)$ & $(\eps N^{-i},0,0,0.8)$ \\
$b_j\in B$ & $(0,\eps N^{-j},0,1)$ & $(0,\eps N^{-j},0,0.9)$\\
$c_k\in C$ & $(0,0,\eps N^{-k},1)$ & $(0,0,\eps N^{-k},1.3)$\\
\end{tabular}
\end{center}
In \cite{BhaskaraKTW13} the authors show that there is a schedule with makespan at most $3 + 3\eps$ if there is a 3D-matching.
Furthermore, two Lemmata (Lemma 2.1 and 2.2 in \cite{BhaskaraKTW13}) are used to show that the existence of a schedule with makespan at most $3.09 + 3\eps $ implies the existence of a 3D-matching. 
We present a counter example.
Let $n = 3$ and:
\[E = \set[\big]{\set{a_1, b_1, c_2}, \set{a_2, b_2, c_2}, \set{a_3, b_3, c_3}, \set{a_3, b_2, c_3} , \set{a_3, b_3, c_1}}\]
Hence, we have:
\begin{center}
\begin{tabular}{l|lllllllll}
$x$     & $a_1$ & $a_2$ & $a_3$ & $b_1$ & $b_2$ & $b_3$ & $c_1$ & $c_2$ & $c_3$ \\\midrule
$|E(x)|$& 1     & 1     & 3     & 1     & 2     & 2     & 1     & 2     & 2
\end{tabular}
\end{center}

In order to match $a_1$ and $a_2$, a 3D-matching would have to contain the first two triplets matching $c_2$ twice.
Hence, there is no 3D-matching for this instance.

On the other hand, we define a schedule with makespan $3 + 3\eps\leq 3.09 + 3\eps$.
\begin{center}
\begin{tabularx}{0.8\textwidth}{X|XXXXX}
Machine & $\set{a_1, b_1, c_2}$ & $\set{a_2, b_2, c_2}$ & $\set{a_3, b_3, c_3}$ & $\set{a_3, b_2, c_3}$ & $\set{a_3, b_3, c_1}$ \\\midrule
Jobs & $a_1$, $a_2$, $b_1$ (element) & $a_3$, $b_2$, $c_2$ (dummy) & $a_3$, $b_3$, $c_3$ (dummy) & $a_3$, $b_2$, $c_3$ (element) & $b_3$, $c_1$, $c_2$ (element) \\\midrule
Load & $3 + (2 + 1/N)\eps$ & $3 + (2 + 1/N)\eps$ & $3 + 3\eps$& $3 + 3\eps$ & $3 + (2 + 1/N)\eps$
\end{tabularx}
\end{center}

\end{document}